\documentclass
[journal,comsoc]{IEEEtran}
\usepackage{algorithm,amsbsy,amsmath,amssymb,epsfig,bbm,mathrsfs,multirow,amsthm}
\usepackage{color}
\usepackage{algpseudocode}
\usepackage{graphicx,caption,subcaption,array,multirow}%for three figures side-by-side
\usepackage{makecell}
\usepackage{tabularx}
\usepackage{cite}

\DeclareMathAlphabet{\mathpzc}{OT1}{pzc}{m}{it}

\newtheorem{theorem}{\textbf{\textsc{Theorem}}}

\begin{document}
\title{Securing MIMO Wiretap Channel with Learning-Based Friendly Jamming under\\ Imperfect CSI}
\author{
\IEEEauthorblockN{Bui Minh Tuan, Diep N. Nguyen, Nguyen Linh Trung, Van-Dinh Nguyen, Nguyen Van Huynh, Dinh Thai Hoang, Marwan Krunz,  and Eryk Dutkiewicz
\thanks{Bui Minh Tuan, Diep N. Nguyen, Dinh Thai Hoang, and Eryk Dutkiewicz are with the School of Electrical and Data Engineering, University of Technology Sydney, NSW 2007, Australia (bui.m.tuan@student.uts.edu.au; \{diep.nguyen, hoang.dinh, eryk.dutkiewicz\} @uts.edu.au).}
}
\thanks{Van-Dinh Nguyen is with the College of Engineering and Computer Science and the Center for Environmental Intelligence, VinUniversity, Vinhomes Ocean Park, Hanoi, Vietnam (dinh.nv2@vinuni.edu.vn).}
\thanks{Nguyen Van Huynh is with the Department of Electrical Engineering and Electronics, University of Liverpool, Liverpool, L69 3GJ, United Kingdom (e-mail: huynh.nguyen@liverpool.ac.uk).}
\thanks{Marwan Krunz is with the Department of Electrical and Computer Engineering, The University of Arizona, Tucson, AZ, USA (krunz@arizona.edu).}
\thanks{Nguyen Linh Trung (correspondence) is with the Advanced Institute of Engineering and Technology, University of Engineering and Technology, Vietnam National University, Hanoi, Vietnam (linhtrung@vnu.edu.vn).}
}
\maketitle

\begin{abstract}
Wireless communications are particularly vulnerable to eavesdropping attacks due to their broadcast nature. To effectively deal with eavesdroppers, existing security techniques usually require accurate channel state information (CSI), e.g. for friendly jamming (FJ), and/or additional computing resources at transceivers, e.g. cryptography-based solutions, which unfortunately may not be feasible in practice. This challenge is even more acute in low-end IoT devices. We thus introduce a novel deep learning-based FJ framework that can effectively defeat eavesdropping attacks with imperfect CSI and even without CSI of legitimate channels. In particular, we first develop an autoencoder-based communication architecture with FJ, namely AEFJ, to jointly maximize the secrecy rate and minimize the block error rate at the receiver without requiring perfect CSI of the legitimate channels. In addition, to deal with the case without CSI, we leverage the mutual information neural estimation (MINE) concept and design a MINE-based FJ scheme that can achieve comparable security performance to the conventional FJ methods that require perfect CSI. Extensive simulations in a multiple-input multiple-output (MIMO) system demonstrate that our proposed solution can effectively deal with eavesdropping attacks in various settings. Moreover, the proposed framework can seamlessly integrate MIMO security and detection tasks into a unified end-to-end learning process. This integrated approach can significantly maximize the throughput and minimize the block error rate, offering a good solution for enhancing communication security in wireless communication systems.
\end{abstract}
	
\section{Introduction}
The rapid development of the Internet of Things (IoT) has paved the way for establishing connections and gathering data across various domains, from smart homes and cities to healthcare, transportation, and Industry 4.0. However, with a large number of diverse wireless devices, IoT networks are highly vulnerable to adversaries due to the broadcast nature of wireless communications. Among radio threats, eavesdropping is one of the most common and serious attacks. Specifically, with off-the-shelf circuits, adversaries can effectively launch eavesdropping attacks to passively wiretap legitimate wireless channels to intercept the transmitted information from the transmitter. As a result, the need for lightweight anti-eavesdropping approaches that accommodate the limited processing capabilities of IoT devices without compromising their operational efficiency is, therefore, critical.

\subsection{Related Work}
The most common countermeasure against eavesdroppers is encryption techniques \cite{mukherjee2014principles}, such as Advanced Encryption Standard (AES) and Rivest-Shamir-Adelman (RSA). Unfortunately, these methods require infrastructure for distributing public keys and the irresistibility of the underlying encrypted function. In addition, data encryption and decryption create burdens on computing resources for wireless devices, especially for computing- and resource-constrained devices such as IoT devices and wireless sensors~\cite{zhang2015secure, chu2023countering}. Moreover, conventional encryption techniques may not guarantee absolute protection when eavesdroppers have sufficient computational capability or employ sophisticated techniques to crack encryption keys~\cite{mukherjee2014principles}. Hence, when transceivers have limited capabilities as compared to eavesdroppers, such as IoT devices, legitimate communication links are still subject to being compromised. Finally, quantum computing may soon crack encryption techniques while post-quantum cryptography is still in its infancy and may not be feasible/suitable for power-constrained devices~\cite{bernstein2017post}.

Given the limitations of encryption techniques, physical layer security (PLS) has emerged as a promising solution to deal with eavesdropping attacks by leveraging the dynamic and stochastic nature of wireless channels~\cite{bloch2008wireless}. Specifically, PLS aims to achieve a positive secrecy capacity defined by the difference between the channel capacities of the legitimate link and the eavesdropper link~\cite{shannon1949communication,wyner1975wire}. To do that, PLS can focus on either random key generation or provisioning information-theoretic secrecy rates to design secure wireless communication systems~\cite{hamamreh2018classifications}. The former extracts channel state information (CSI) to generate random keys that can be used for authentication purposes. The latter includes channel coding, channel-based adaptation, and artificial noise (AN)/friendly jamming (FJ)-based security. Among these PLS approaches, the FJ technique has been widely studied in the literature due to its effectiveness in dealing with eavesdropping attacks. With the FJ technique, the transmitter (Tx) uses a proportion of the transmit power to deliberately inject noise signals to degrade the reception of the legitimate signals at eavesdroppers while not adversely impacting the legitimate receiver (Rx) \cite{wyner1975wire}. This can be achieved with the beamforming technique that is popular in multi-input, multi-output or multi-input-multi-output (MIMO) systems. The critical factor affecting the security performance of FJ is the availability of CSI at the Tx for beamforming purposes. In one of the early studies on FJ-based security~\cite{goel2008guaranteeing}, Tx-based friendly jamming (TxFJ) and cooperative jamming were proposed with the assumption that the perfect CSI between legitimate transceivers is known at the Tx. The FJ signal is then designed to project on the nullspace/directions of the legitimate channel so that the FJ noise signal will not affect the signal-to-noise ratio (SNR) of the Rx. Rx-based friendly jamming (RxFJ) was also proposed with a multi-user broadcast channel in~\cite{akgun2016exploiting}, guaranteeing a non-zero average secrecy rate, regardless of the position of the eavesdropper.

Securing multi-user link while managing the adverse effect of the FJ interference was further studied in~\cite {siyari2017friendly}. The authors used a non-cooperative game to model the FJ and transmit power control problem for the interfering wiretap links. In addition, the work in~\cite{choi2015robust} investigated beamforming based-FJ to achieve an instantaneous secrecy rate. In particular, the authors first derived the closed-form secrecy rate. Then, the Tx picks a suitable coding scheme to guarantee secrecy. Nevertheless, all the above methods as well as most general MIMO settings e.g., \cite{goel2008guaranteeing,akgun2016exploiting,siyari2017friendly,MIMO1, MIMO2,MIMOML} required or assumed perfect CSI or, at least, statistical CSI of the legitimate channel at the Tx to construct FJ signals. However, perfect or statistical CSI acquisition is very challenging, if not impractical, especially for IoT communications. 

The key reasons are the inaccuracy of the estimated CSI at the Rx (before feeding it back to the Tx) and the limited bandwidth of the feedback channel~\cite{mukherjee2015physical}. Hence, the FJ signals, designed to lie on the nullspace of the imperfect channel \cite{goel2008guaranteeing, akgun2016exploiting}, may degrade the quality of the (actual) legitimate channel and not be fully cancelled out at the Rx.
To overcome these issues, the authors in \cite{choi2015robust,mukherjee2010robust, 6101597} assumed statistical CSI (instead of full/complete CSI) to satisfy the QoS of SNR and the nonzero secrecy capacity. Then, the second perturbation analysis was used to design the beamforming scheme. Unfortunately, this method requires high computational capability at the Rx and high running time. 
Another approach in~\cite{6101597} proposed a masked beamforming scheme that leverages artificial noise to impair the decoding capabilities the eavesdropper, particularly regarding channel estimation delay and error. The instantaneous security performance of this scheme was then evaluated in terms of the minimum mean square error at the eavesdropper, providing a metric to assess its effectiveness in enhancing communication security amidst the discussed challenges. However, the alternative optimization used in the proposed approach has a low convergence speed. Furthermore, the work in~\cite{9319238} introduced an FJ-based authentication scheme. This scheme leverages FJ to obscure a hash-based message authentication code tag, facilitating secure authentication between legitimate users. This method requires significant overhead to be added to the current communication infrastructure.

%Due to the limitations of encryption techniques, PLS has emerged as a promising solution to combat eavesdropping by leveraging wireless channels' dynamic and stochastic nature \cite{mukherjee2014principles}. Specifically, PLS aims to achieve a positive secrecy capacity defined by the difference between the channel capacities of the legitimate link and the eavesdropper link \cite{wyner1975wire}. as the difference between the channel capacities of the legitimate and eavesdropping links \cite{mukherjee2014principles}.To achieve this, PLS can focus on either random key generation or providing information-theoretic secrecy rates to design secure wireless systems. Random key generation extracts channel state information (CSI) to create keys for authentication \cite{hamamreh2018classifications}.  while information-theoretic approaches include techniques like channel coding, channel adaptation, and security methods based on artificial noise (AN) or friendly jamming (FJ). Among these, the FJ technique is widely studied for its effectiveness against eavesdropping. In this approach, the transmitter (Tx) dedicates a portion of its power to injecting noise signals that degrade the eavesdropper’s reception without negatively impacting the legitimate receiver (Rx) \cite{goel2008guaranteeing}. 

\captionsetup{justification=centerlast}

\begin{table*}[ht]
\centering
\caption{Summary of Recent Studies on Friendly Jamming-based PLS.}
\label{tab:Summary_Recent_Studies}
\resizebox{\textwidth}{!}{ % Resize to fit within text width
\begin{tabular}{|p{3cm}|p{7cm}|p{2.5cm}|p{3cm}|p{1.5cm}|}
\hline
\textbf{Related Works} & \textbf{Methodology} & \textbf{CSI} & \textbf{SC Optimization} & \textbf{BLER Optimization} \\
\hline
\multicolumn{5}{|l|}{\textbf{Conventional Approaches}} \\
\hline
Goel, \textit{et al.} \cite{goel2008guaranteeing}; B. Akgun, \textit{et al.} \cite{akgun2016exploiting}; Tsai, \textit{et al.} \cite{Tsai_Power_Allocation}; Hu, \textit{et al.} \cite{Hu_Co_Jamming_PLS} & FJ signal is precoded in the null-space of the legitimate channel while the information signal is in the range space of the channel & Perfect CSI & Yes, transmit power allocation based on exhaustive search & No \\
\hline
Siyari, \textit{et al.} \cite{siyari2017friendly} & Non-cooperative Game Theory & Perfect CSI & Yes & No \\
\hline
J. Choi \cite{choi2015robust} & Masked beamforming & Perfect CSI & Yes, semidefinite programming (SDP) & No \\
\hline
Mukherjee, \textit{et al.} \cite{mukherjee2010robust} & Null-space FJ beamforming & Statistical CSI & Yes & No \\
\hline
\multicolumn{5}{|l|}{\textbf{ML/DL Approaches}} \\
\hline
Yun, \textit{et al.} \cite{8957321} & Using both supervised pre-training to optimize the secrecy rate & Statistical CSI & Yes & No \\
\hline
Gou, \textit{et al.} \cite{Guo_Proactive_multiagent} & Neural fictitious self-play with soft actor-critic (NFSP-SAC), Multiagent Deep Reinforcement Learning-based FJ & Perfect CSI on legitimate channel & Yes & No \\
\hline
Our work & AE-based FJ and MINE-based FJ & Statistical CSI & Yes, leverage end-to-end learning and MINE to optimize SC & Yes \\
\hline
\end{tabular}
} % End of resizebox
\end{table*}

Table 1 summarizes conventional FJ approaches. In [1], the authors proposed transmitter-based friendly jamming (TxFJ) and cooperative jamming, assuming perfect CSI is available at Tx. The FJ signal is precoded in the null space of the legitimate channel, ensuring that the noise does not degrade the signal-to-noise ratio (SNR) at Rx, thus maintaining a non-zero average secrecy rate regardless of the eavesdropper's position. However, optimizing the secrecy rate involves an exhaustive search, leading to high computational costs. Another approach, based on game theory \cite{siyari2017friendly}, uses both Tx and Rx FJ, with power allocation optimized through a non-cooperative game model. The masked beamforming designed the FJ signal to be orthogonal to the transmit signal. Still, like other methods, it requires perfect CSI of the legitimate channel, which is difficult to achieve, especially in IoT communications. To address these limitations, the authors in \cite{mukherjee2010robust} proposed a beamforming-based FJ scheme that meets quality of service (QoS) requirements for SNR and guarantees non-zero secrecy capacity. However, designing the beamforming scheme using second-order perturbation analysis demands high computational resources at the receiver, making it less practical due to increased processing time and complexity.

Regarding DL/ML-based FJ approaches, Yun \textit{et al.} \cite{8957321} developed a learning-based FJ method that combines supervised pre-training with unsupervised post-training to optimize the secrecy rate. This approach accounts for imperfect CSI on the legitimate channel, utilizing available CSI error statistics. However, the two-stage training process adds overhead and may be impractical for real-world applications. In contrast, our work addresses imperfect CSI in more practical scenarios, ranging from statistical CSI to cases where no CSI is available at the transmitter. During offline training, outdated CSI containing errors is used to train the model. This allows the model to learn the imperfect channel conditions while simultaneously optimizing secrecy capacity (SC). In the testing stage, the model continues to achieve optimal SC even with new input and updated CSI.

\subsection{Motivation and Main Contributions}
As mentioned earlier, the security performance of FJ-based methods is highly affected by imperfect CSI (ImCSI). To cope with it, conventional methods were based on exhaustive search and second-order statistics that require high computational and transmission overhead. This work aims to develop an FJ solution based on deep learning (DL) in order to deal with ImCSI. To this end, we integrate the autoencoder (AE) architecture, with denoising and generalization capability, into the FJ approach. AE is designed to acquire a dataset representation (encoding) and reconstruct (decoding) the input from the output. Based on this encoding, it reconstructs a representation of the output that closely resembles the input \cite{erpek2018learning, o2017deep}. Unlike the previous works in ~\cite{mukherjee2010robust,fritschek2019deep, pmlr-v80-belghazi18a,9449919, Fulwani_DL_MISO}, and~\cite{lin2019beamforming}, we introduce a novel FJ method for PLS utilizing AE while the statistical information of channel estimation error is known. 
% as outlined in~\cite{tuan2020autoencoder}.
% Hence, we can address this issue by integrating AE into the FJ approach.
Consequently, the running time and overhead of our DL-based method can be reduced significantly compared to the conventional methods.

Moreover, to deal with the case without even statistical CSI (SCSI) at the transmitter, we integrate the method of mutual information neural estimation (MINE)~\cite{pmlr-v80-belghazi18a} into our AE architecture. MINE has been employed for transceiver design in a few studies in the literature. For example, the authors in~\cite {Fritschek2019} adopted MINE to improve the channel coding process at the transmitter. Differently, the authors in~\cite{9449919} demonstrated that combining AE and MINE could enhance transmission performance, including decoding error and throughput. However, these works and others in the literature only considered single-input-single-output systems and may not be effective when dealing with eavesdropping attacks. Our study here expands upon the traditional concept of relying on perfect CSI  to cancel out FJ at the Rx. In particular, our method empowers the legitimate  Tx and Rx to learn how to suppress FJ signals while degrading the eavesdropper~channel. In addition, the related works discussed above, including both conventional FJ and deep learning-based FJ, have not considered the BLER, or the reliability of the legitimate transmission, while optimizing the secrecy rate under imperfect CSI, from statistical to unknown CSI. Our work fills this gap. Additionally, it is worth noting that the proposed Deep Learning techniques offer a unified and scalable solution that can be applicable to both conventional optimization and existing learning-based approaches. The experimental results demonstrate the effectiveness of our proposed solutions in addressing the trade-offs between secrecy performance and computational complexity, especially in IoT and resource-constrained environments. In short, the major contributions of the paper are summarized as follows:

\begin{itemize}
    \item We jointly consider the security and reliability of the legitimate transmission by proposing an end-to-end learning-based framework that maximizing the secrecy rate while not sacrifying the reliability of the transmission, captured by the BLER.
    
    \item Our method can accommodate different levels of CSI, from perfect to statistical CSI or even completely unknown CSI, making it practical for real-world communication scenarios where CSI accuracy is often limited.
    
    \item We introduce a unified framework that combines end-to-end learning and secrecy rate optimization, integrating beamforming, jamming signal generation, and learning-based adaptation. This framework maximizes the signal-to-noise ratio (SNR) at the legitimate receiver while dynamically optimizing the secrecy capacity in real-time based on channel conditions.
    \item Extensive simulations are conducted under various CSI conditions, demonstrating that our proposed learning-based FJ model outperforms conventional methods in terms of secrecy capacity, BLER, and overall system efficiency, making it a practical and effective solution for secure communication in practical environments with imperfect or even no CSI.
\end{itemize}

The structure of this paper is as follows: Section~\ref{sec:FJbackground} provides a brief overview of the system model and the FJ-based Physical Layer Security (PLS) technique as a reference point for our new approach. Following this, Section~\ref{sec:AEFJ} details our novel AE-based FJ strategy, designed to counter eavesdroppers when only statistical information of channel estimation error is available. Section~\ref{subsubsec:MINE-based FJ} introduces our security framework utilizing MINE to safeguard communications without CSI. The analysis of secrecy and Block Error Rate (BLER) by simulations is presented in Section~\ref{sec:sim}. Conclusions are drawn in Section \ref{sec:VII}.
	
{\bfseries Notation:} Vectors and matrices are denoted by bold lowercase and uppercase letters, respectively. The absolute value of a real number, the magnitude of a complex number, and the complex conjugate transpose are, respectively, denoted by $\left\|\cdot\right\|$, $\left| \cdot\right|$ and $(\cdot)^\dag$. The complex Gaussian random variable with mean $ \mu $ and variance $ \sigma^2$ is denoted  by $\mathcal{CN}(\mu ,\sigma^2 )$. $\mathbb{E}[X]$ denotes the expectation of $X$, $\text{Tr} (\textbf{A})$ the trace of the square matrix $\textbf{A}$. 
	
\section{System Model and Secrecy Rate Optimization}
\label{sec:FJbackground}

\begin{figure*}[!]   
\centering
\includegraphics[width=0.68\linewidth]{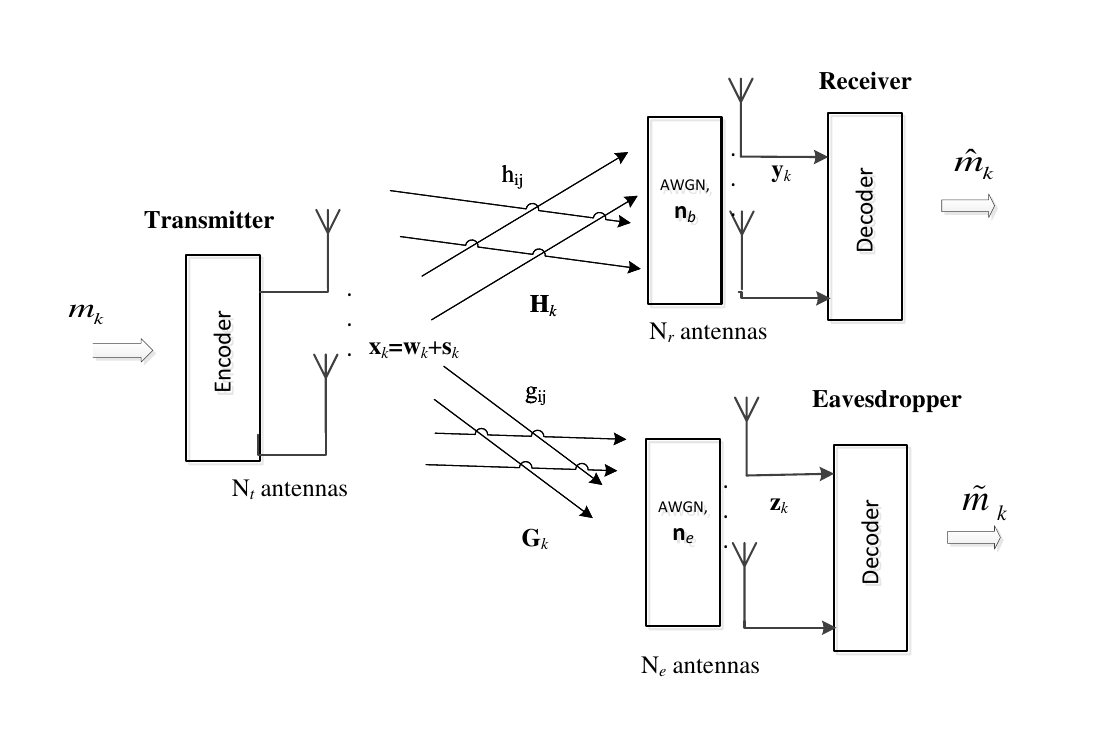}
% \\[-.7cm]
\caption{Illustration of MIMO FJ-based system model with Tx, Rx, and eavesdropper.}
\label{fig: MIMO_FJ_Conven}
\end{figure*}

\begin{figure*}[!]   
    \centering
    \includegraphics[width=0.8\linewidth]{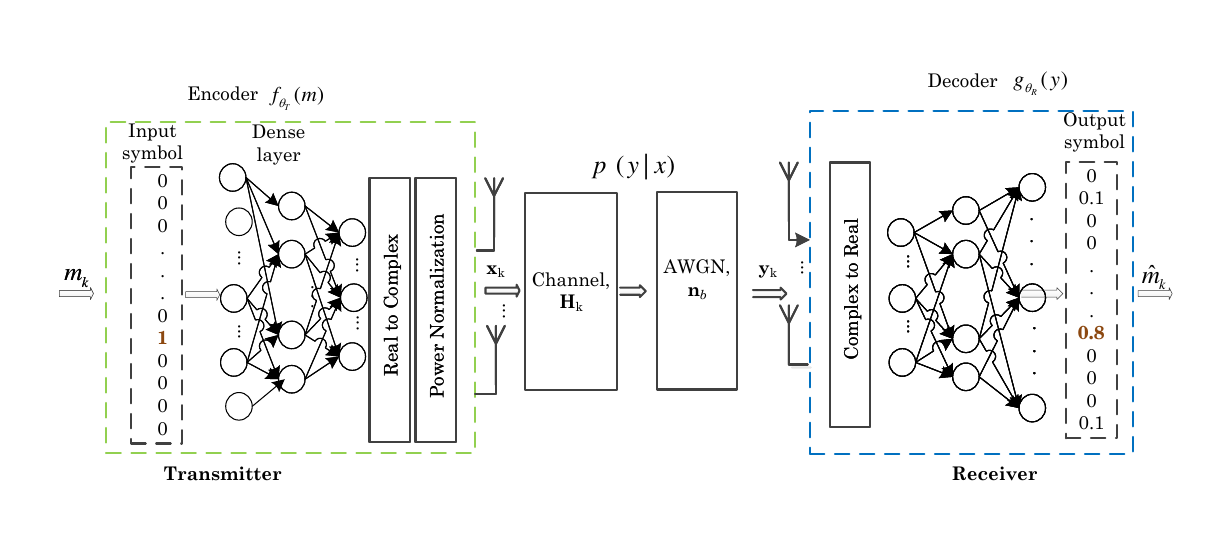}
    % \\[-1cm]
    \caption{The AE-based communication with E2E learning.}
    \label{fig:AE_based_comms}
\end{figure*} 

\subsection{Background on MIMO-FJ}
\subsubsection{System Model}
The standard MIMO FJ model is depicted in Fig. \ref{fig: MIMO_FJ_Conven}, where the Tx, the Rx, and the eavesdropper are equipped with $N_t$, $N_r$ and $N_e$ antennas, respectively~\cite{goel2008guaranteeing}. Channel matrices at time slot $k$ for the Tx-Rx and Tx-eavesdropper links are denoted as $\mathbf {H}_k \in \mathbb{C}^{N_t \text {x} N_r}$ and $\mathbf {G}_k\in \mathbb{C}^{N_t \text {x} N_e}$.
Each element of $\mathbf {H}_k$ and $\mathbf {G}_k$ is assumed to be known, i.i.d., and unchanged over a block of transmit symbols. Let $\mathbf{s}_k$ be the desired message to the Rx. To cancel interference caused by FJ at the Rx, the FJ signal $\mathbf{w}_k$ is designed to satisfy $ {\mathbf {H}_k^\dag \mathbf{w}_k = 0}$. Given the transmitted signal $\mathbf{x}_k = \mathbf{s}_k + \mathbf{w}_k$, the received signals at the Rx and the eavesdropper can be expressed as follows:
\begin{align}
    \label{eq:ABE_Rela}
    \mathbf{y}_k &= \mathbf {H}_k^\dag \mathbf{x}_k+\mathbf{n}_b 
    = \mathbf {H}_k^\dag \mathbf{s}_k + \mathbf {H}_k^\dag \mathbf{w}_k+\mathbf{n}_b,\\
    \mathbf{z}_k &= \mathbf {G}_k^\dag \mathbf{s}_k + \mathbf {G}_k^\dag \mathbf{w}_k+\mathbf{n}_e
\end{align}
where $\mathbf{n}_b\sim\mathcal{CN}(0,\sigma^2_b \mathbf{}{I})$ and $\mathbf{n}_e\sim\mathcal{CN}(0,\sigma^2_e \mathbf{}{I})$ are the additive white Gaussian noise (AWGN) at the Rx and the eavesdropper. FJ signal $\mathbf{w}_k$ is designed as $\mathbf{w}_k=\mathbf{Z}_k\mathbf{v}_k$, where the precoding matrix $\mathbf{Z}_k$ is the orthogonal basis of the nullspace of $\mathbf {H}_k ^ \dag$. The elements of $\mathbf{v}_k$ are i.i.d. complex Gaussian random variables with variance $\sigma_v^2$. The noise covariance at the eavesdropper then can be calculated as follows:
	\begin{align}
        \label{eq:Cov_noise_Eve}
        \mathbf{K}_k = (\mathbf {G}_k ^\dag\mathbf {Z}_k ^\dag\mathbf {Z}_k \mathbf{G}_k)\sigma_v^2 + \mathbf I_{N_e}\sigma_e^2,
	\end{align}
where $\textbf{I}_{i}$ is an $i \times i$ identity matrix.  
\subsubsection{The Secrecy Rate Optimization Problem} 
Given $R_{AB}$ and $R_{AE}$ as the legitimate and illegitimate channel rates, the secrecy rate $R_k^s$ (nats/s/Hz) is given by
	\begin{align}   
		R_k^s &\triangleq~ \left[R_{AB} - R_{AE}\right]_+ %\notag\\
		= \left[ \log (1 + \operatorname{SINR}_\text{B}) - \log (1 + \operatorname{SINR}_\text{E})\right]_+ \notag \\
		&= \left[\log \det\big(\mathbf{I} + \mathbf {H}_k^\dag \mathbf {Q}_s\mathbf {H}_k\big) 
         - \log\frac{{{{\det \big(\mathbf {K}_k+\mathbf {G}_k^\dag\mathbf {Q}_s\mathbf {G}_k}\big)}}}{{{\det(\mathbf {K}_k)}} }\right]_+
        \label{eq:MIMOFJ}
	\end{align}
where $[x]_+ = \max(0,x)$, $\mathbf{Q}_s=\mathbf{E}[\mathbf{s}_k\mathbf{s}_k^\dag]$. Since the eavesdropper's CSI is unavailable at the Tx, we focus on maximizing the first term in~\eqref{eq:MIMOFJ}  by using singular value decomposition (SVD). Specifically, the legitimate channel $\mathbf {H}_k$ is first decomposed as
\begin{align}   
    \mathbf{H}_k ^ \dag &= \mathbf {U}_k \mathbf {\Gamma}_k \mathbf {V}_k^\dag,
\end{align}
where $\mathbf {U}_k \in \mathbb{C}^{N_r \text {x} N_r}$, $\mathbf {V}_k \in \mathbb{C}^{N_t \text {x} N_t}$, and $\mathbf {\Gamma}_k \in \mathbb{C}^{N_t \text {x} N_t}$. After precoding $\mathbf {r}_k=\mathbf {V}_k^\dag\mathbf {s}_k$, the received signal $\mathbf {y}_k$ is multiplied with $\mathbf{U}_k^\dag$. Then, the received signal in \eqref{eq:ABE_Rela} can be equivalently to
\begin{align}   
    \mathbf{\tilde {y}}_k &= \mathbf {\Gamma}_k^\dag \mathbf{r}_k +\mathbf{\tilde{n}}_b,
\end{align}
where $\mathbf{\tilde {y}}_k=\mathbf {U}_k^\dag\mathbf {y}_k$. We denote $P$ is the total transmit power, $P_{\mathrm{info}}$ is the power for transmitting information signal, and $\mathbf {Q}_r$ as the precoded transmit covariance. Then $\mathbf {Q}_r$ is designed as

% Your text here
Then $\mathbf{Q}_r$ is designed as
\begin{align}
\mathbf{Q}_r = \mathbb{E} \left[ \mathbf{r}_k \mathbf{r}_k^{\dagger} \right] 
= \operatorname{diag}(\sigma_{r,1}^2, \sigma_{r,2}^2, \dots, \sigma_{r,N_t}^2),
\end{align}

where $\sigma_{r,i}^2$ is derived by the water filling solution with power constraint $P_{\mathrm{info}}\leq {P}$, corresponding to the $N_t$ largest singular values of $\mathbf{H}_k ^ \dag$.  Thus, the secrecy rate $R_k^s$ is re-expressed as
\begin{align}   
    R_k^s= \left[ \log \det (\mathbf{I} + {\mathbf {\Gamma}_k\mathbf {Q}_r\mathbf {\Gamma}_k^\dag}) 
    - \log\frac{\det (\mathbf {K}_k+\mathbf {F}_k)}{\det(\mathbf {K}_k)} \right]_+,
    \label{eq:MIMOFJ_Op}
\end{align}
where $\mathbf {F}_k=\mathbf {G}_k^\dag \mathbf {V}_k^\dag \mathbf {Q}_r \mathbf {V}_k \mathbf {G}_k$.
Since $R_k^s$ is random, the average secrecy rate over a number of channel realizations $\bar{R}$ will be used. Our objective is to maximize the average secrecy rate subject to the power constraint at the Tx, which is mathematically formulated as follows: 
\begin{equation}
\label{eq:Av_Cs_MIMO}
    \begin{aligned} 
     \bar{R}~\dot= \max_{\text{Tr}(\mathbb{E}[\mathbf{x}_k\mathbf{x}_k^\dag]) \le {P}} \mathbb{E} \left[\log \det (\mathbf{I} + {\mathbf {\Gamma}_k\mathbf {Q}_r\mathbf {\Gamma}_k^\dag}) - \log\frac{\det (\mathbf {K}_k+\mathbf {F}_k)}{\det(\mathbf {K}_k)} \right].
   \end{aligned}
\end{equation}
The power constraint $\text{Tr}(\mathbb{E}[\mathbf{x}_k\mathbf{x}_k^\dag]) \le {P}$ can be rewritten as $ \text{Tr}\big(\mathbf {V}_k^\dag \mathbf{Q}_r\mathbf{V}_k+N_{\mathrm{FJ}}\sigma_v^2 {\mathbf{I}_k}_{N_t} \big)\le P$, where ${N}_{\mathrm{FJ}}$ denotes the number of dimensions used for transmitting FJ signals.
% \commentTuan{edit the worst case}

\subsubsection{Channel Estimation Error Model}
The previous part considers the case of perfect CSI. Regarding ImCSI, due to channel estimation error, we choose a common ImCSI model~\cite{mukherjee2010robust}:
\begin{align}
    {{\rm{\tilde {\mathbf{H}}}}_k}{\rm{ = }}{{\rm{\mathbf{H}}}_k} + \Delta {{\rm{\mathbf{H}}}_k},
\end{align}
where the channel estimation error $\Delta {\mathbf{H}_k}$ is modeled as an i.i.d. complex Gaussian and circularly symmetric random matrix with the covariance ${\mathbf{C}_{\Delta \mathbf{H}_k}} \sim \mathcal{CN}({\mathbf{0}},\rho_e^2{\mathbf{I}}_{N_r})$. 
Then the secrecy rate in~\eqref{eq:MIMOFJ} becomes
\begin{align}
    \overline R_k^s \dot= \max_{\text{Tr}(\mathbb{E}[\mathbf{x}_k^\dag\mathbf{x}_k]) \le {P}} \mathbb{E}&\Big[\log \frac{{\det ({\bf{I}} + {{({{\bf{H}}_k} + \Delta {{\bf{H}}_k})}^\dag }{{\bf{Q}}_s}({{\bf{H}}_{\bf{k}}} + {\bf{\Delta }}{{\bf{H}}_{\bf{k}}})}}{{\det ({{\bf{D}}_k})}} \notag\\  &-\log\frac{{{{\det (\mathbf {K}_k+\mathbf {G}_k^\dag\mathbf {Q}_s\mathbf {G}_k})}}}{{{\det(\mathbf {K}_k)}} } \Big], \label{eq:MIMOFJ-ImCSI}
\end{align}
where $\mathbf{D}_k = (\mathbf {\Delta \mathbf{H}_k ^\dag\mathbf {Z}_k ^\dag\mathbf {Z}_k \mathbf{\Delta \mathbf{H}}_k)\{sigma_v^2 + \mathbf I_{N_r}\sigma_b^2}$ is the noise covariance at the receiver. The reason is the FJ signals are not completely cancelled out at the receiver due to channel error.
\subsection{AE-based MIMO Communication}
The AE-based communication is illustrated in Fig. \ref{fig:AE_based_comms} \cite{o2017introduction}. The model aims to reconstruct the input messages at time $k$,  $m_k\in \mathcal{M} = \{ 1,2,\ldots, M\}$, by minimizing the reconstruction error of the inputs or transmitted signals.
The message $m_k$ is generated uniformly and then embedded into the transmitted symbol $\mathbf{s}_k$. Then, it is sent through the dense layers to generate the transmit signal $\mathbf{x}_k$. The batch normalization layer performs the power constraint. The decoder includes the complex to real conversion layer, dense layers, and the last layer using the Softmax activation function. The output is the probability distribution $ \hat{\mathbbm{1}}_m \in (0,1)^{\operatorname{card}(M)} $ over all messages ($\operatorname{card}$ denotes cardinality). Then, the maximum likelihood (ML) is used to estimate the transmitted signal and optimize the construction error by using the cross-entropy loss function ~\cite{o2017introduction,bengio2017deep}. Hence, the decoded message $\hat{m}_k$ will be the index of the element of $\hat{\mathbbm{1}}_m $ with the highest value. 

The training samples in our model include transmitting messages as an input and the encoder, labels at the decoder, channel realizations and AWGN at both legitimate and illegitimate channels. Firstly, the transmitted messages consist of $M$ classes of integers and $L$ labelled instances, denoted as $\mathbf{m}_k$ and $\mathbf{l}_k$, respectively, where $\mathbf{l}_k \in M = \{1, \ldots, M\}$ is a class label of input $\mathbf{m}_k$. Secondly, the channel samples $\textbf{H}_k \sim \mathcal{CN}^{{N_t} \times {N_R}}$, contains the elements $h_{ij} = a_{ij} + jb_{ij}$, where $a_{ij}$, $b_{ij} \sim \mathcal{CN}(0,1)$. This is applied to the illegitimate channel, AWGN at the receiver and eavesdropper, $\textbf{G}_k$, $\mathbf{n}_b$, and $\mathbf{n}_e$, respectively. Since the AE and MINE-based FJ methods operate only with real parameters, the complex samples are parameterized to produce real values as follows:
\begin{align}
&{{\rm{\hat {\mathbf{x}}}}_k}{\rm{ = }}\left[ {\begin{array}{*{20}{c}}
{{{\Re}}({{\rm{\mathbf{x}}}_k})}\\
{{\Im}({{\rm{\mathbf{x}}}_k})}
\end{array}} \right],~{{\rm{\hat {\mathbf{y}}}}_k}{\rm{ = }}\left[ {\begin{array}{*{20}{c}}
{{{\Re}}({{\rm{\mathbf{y}}}_k})}\\
{{\Im}({{\rm{\mathbf{y}}}_k})}
\end{array}} \right],~
{{\rm{\hat {\mathbf{n}}}}_b}{\rm{ = }}\left[ {\begin{array}{*{20}{c}}
{{{\Re}}({{\rm{\mathbf{n}}}_b})} \notag \\
{{\Im}({{\rm{\mathbf{n}}}_b})}
\end{array}} \right],\\&{\rm{\hat {\mathbf{H}}_k = }}\left[ {\begin{array}{*{20}{c}}
{{\Re}({\rm{\mathbf {H}_k}})}&{ - {\Im}({\rm{\mathbf {H}_k}})}\\
{{\Im}({\rm{\mathbf {H}_k}})}&{{\Re}({\rm{\mathbf {H}_k}})}
\end{array}} \right],
{{\rm{\hat {\mathbf{n}}}}_e}{\rm{ = }}\left[ {\begin{array}{*{20}{c}}
{{{\Re}}({{\rm{\mathbf{n}}}_e})} \\
{{\Im}({{\rm{\mathbf{n}}}_e})}
\end{array}} \right],\\&{\rm{\hat {\mathbf{G}}_k = }}\left[ {\begin{array}{*{20}{c}}
{{\Re}({\rm{\mathbf {G}_k}})}&{ - {\Im}({\rm{\mathbf {G}_k}})}\\
{{\Im}({\rm{\mathbf {G}_k}})}&{{\Re}({\rm{\mathbf {G}_k}})}
\end{array}} \right] \notag,
\end{align}
where ${\Re}(A)$ and ${\Im}(A)$ are the real and imagine part of A, respectively.

\section{AE-based Friendly Jamming}
\label{sec:AEFJ}
 
In this section, we utilize the AE-based communication model outlined in Fig.~\ref{fig:AE_based_comms} to develop the AE-based FJ approach. The E2E learning scheme aims to achieve two objectives: maximizing the secrecy rate and minimizing the BLER at the receiver. The proposed AE-based MIMO FJ communication and security framework is illustrated in Fig.~\ref{fig:AEFJ}.  
The FJ signal is generated via the FJ generator layer and injected into the transmitted signals $\textbf{x}_k$. 
Regarding channel inputs, For channel inputs, we use estimated channel realizations $\mathbf{\hat{H}}_k$ at both the transmitter and receiver, which are based on outdated data that may include errors $\Delta \mathbf{H}_k$.
During the testing phase, the instantaneous channel realizations differ from those encountered in training but maintain the same distribution.

The proposed AE-based MIMO FJ communication and security scheme is illustrated in Fig.~\ref{fig:AEFJ}. 
The FJ signal is generated via the FJ generator layer and injected into the transmitted signals $\mathbf{x}_k$. 
    \begin{figure*}[t]
        \centering
        \includegraphics[width=0.75\linewidth]{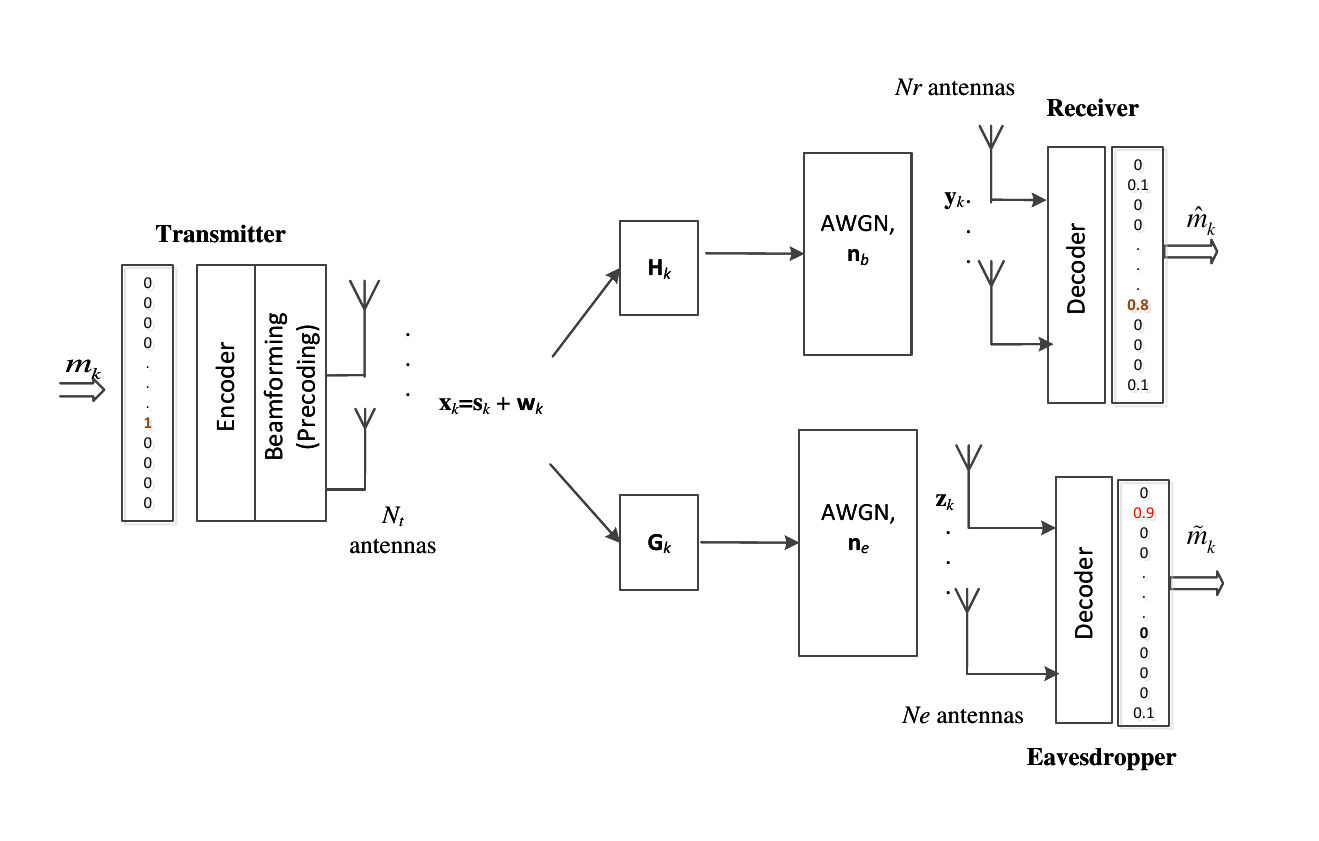}
        % \\[-.5cm]
        \caption{AE-based MIMO FJ.}
        \label{fig:AEFJ}
    \end{figure*}
A straightforward approach to optimize the secrecy rate is to maximize the MI difference between the Tx and the Rx and design an FJ signal to degrade the eavesdropper channel. Next, we will detail this approach. 
    
   % Unlike the LCD-FJ proposed above, the FJ signal will be designed as orthogonal with the transmit signal while the mutual information is maximized via the learning process. The advance of the method is that the BLER is also optimized simultaneously with the secrecy rate.

% \subsection{AE-based Secrecy Rate Optimization}
% \subsubsection{Secrecy Rate Optimization}

\subsection{AE-based MI Optimization}
\label{subsection:Rethinking}
It is worth noting that the communication and security tasks are transformed into the classification with the AE, Softmax activation function, and cross-entropy based security loss function \cite{o2017introduction}. 
The conventional work often assumes that the channel $\mathbf{H}_k$ is perfectly known, based on which the average secrecy rate is achieved by maximizing the channel rate or $I(A, B)$ of the legitimate channel and the precoded FJ signals $\mathbf{w}_k$ are then designed to lie in the nullspace of $\mathbf{H}_k$ so as to degrade the eavesdropper channel with the highest probability. Regarding ImCSI, the FJ signals may not completely lie in the true nullspace of $\mathbf{H}_k$. Thus, finding the optimal FJ signals that can be completely canceled out at the Rx and degrade the decoding capability of the eavesdropper is a nontrivial task; even the statistical information of channel estimation error is available at the Tx. 
Since AE-based communication can be seen as a classifier of transmitted messages and symbols, the MI between the input and the output on the legitimate channel, $I(A, B)$, can be maximized using Softmax and cross-entropy loss functions. 
The transmitted message and its label are denoted by the random variables $\mathbf {m}$ and $\mathbf{l}$, respectively. The training data consists of $M$ classes and $N$ labeled instances denoted as $\{(m_k, l_k)\}_{k=1}^N$, where $l_k \in \mathcal{M} = \{1, ..., M\}$ is a class label of input $m_k$. For the task of reconstruction, let $l_k$ be equal to $m_k$. 
The AE is now a neural classifier $AE (\mathbf {m},  \mathbf {l})$ parameterized by $\phi$. It can be seen that $\mathbf {m}$ and $\mathbf {x}$ have the same probability distribution $p(\mathbf {m})$, similarly with $p(\mathbf {l})$ and $\mathbf {y}$. Thus, the MI between $\mathbf {x}$ and $\mathbf {y}$ can be expressed as: 
    \begin{align}
    I(\mathbf {x}, \mathbf {y}) = I (\mathbf {m}, \mathbf {l}) = \sum_{m_k \in \mathcal{M}} \sum_{l_k \in \mathcal{M}} p(\mathbf {m}, \mathbf {l}) \log \frac{p(\mathbf {m}, \mathbf {l})}{p(\mathbf {m})p(\mathbf {l})}.
    \end{align}
    The MI $I(\mathbf {m}),p(\mathbf {l})$ can be rewritten as follows:
    \begin{align}
        I(\mathbf {m}),p(\mathbf {l}) = {\mathbb{E}_{\mathbf {m}),p(\mathbf {l}}}\log \frac{{p(\mathbf {l}|\mathbf {m})}}{{p(\mathbf {l})}},
    \end{align}
where $\mathbb{E}_{\mathbf {m}, \mathbf {l}}$ is the expectation with respect to the joint distribution $p(\mathbf {m}, \mathbf {l})$. Since  $p(\mathbf {m}, \mathbf {l})$ is not available, we use the variational distribution $q$ to estimate $I(p(\mathbf {m}, \mathbf {l}))$. This can be done as follows~\cite{agakov2004algorithm}:
\begin{align}
    I(p(\mathbf {m}, \mathbf {l})) &= \mathbb{E}_{p(\mathbf {m}, \mathbf {l})} \log \frac{p(\mathbf {l} | \mathbf {m})}{p(\mathbf {l})} 
        = \mathbb{E}_{p(\mathbf {m}, \mathbf {l})} \log \frac{q(\mathbf {l}|\mathbf {m})p(\mathbf {l}|\mathbf {m})}{p(\mathbf {l})q(\mathbf {l}|\mathbf {m})}   \notag \\ 
        &= \mathbb{E}_{\mathbf {m}, \mathbf {l}} \log \frac{q(\mathbf {l}|\mathbf {m})}{p(\mathbf {l})} + \mathbb{E}_{\mathbf {m}, \mathbf {l}} \log \frac{p(\mathbf {m}, \mathbf {l})}{q(\mathbf {m}, \mathbf {l})} - \mathbb{E}_{\mathbf {m}} \log \frac{p(\mathbf {m})}{q(\mathbf {m})} \notag \\ 
        &\geq \mathbb{E}_{\mathbf {m}, \mathbf {l}}\left[ \log \frac{q(\mathbf {m}, \mathbf {l})}{p(\mathbf {m})P(\mathbf {l})}\right].     \label{eq:ieq17}
\end{align}
The inequality \eqref{eq:ieq17} holds since the KL divergence maintains non-negativity. The lower bound is tight when $q(\mathbf {m}, \mathbf {l})$ converges to $p(\mathbf {m}, \mathbf {l})$. However, $q(\mathbf {m}, \mathbf {l})$ has to satisfy the probability constraint axioms such as non-negativity, symmetry, and summary of all probability values equal to one. By~\cite{pmlr-v108-mcallester20a}, $q(\mathbf {m}, \mathbf {l})$ is represented by an unconstrained function $g(\gamma)$, given as
     \begin{align}
    q(\mathbf {m}, \mathbf {l}) = \frac{p(\mathbf {m})p(\mathbf {l})}{\mathbb{E}_{l_i \sim p(\mathbf {l})}\exp f_{\theta}(\mathbf {m}, l_i)} \exp g_{\gamma}(\mathbf {m}, \mathbf {l}).
     \end{align}
 Given a neural network $AE(\mathbf {m}, \mathbf {l})$, the Softmax activation function at the last layer is $ S(AE(\mathbf {m}, \mathbf {l})): \mathbb{R}^M \rightarrow \mathbb{R}^M$ and  defined as:
    \begin{align}
    \label{eq:soft}
    S(AE(\mathbf {m}, \mathbf {l})) = \frac{\exp AE(\mathbf {m}, \mathbf {l})}{\sum_{l_i=1}^{M} \exp n(M)_{l_i}}.
    \end{align}
    Then, the expected cross-entropy $L_{CE}$ loss is given as
    \begin{align}
    L_{CE} = -\mathbb{E}_{\mathbf {m}, \mathbf {l}} \big\{AE(\mathbf {m}, \mathbf {l}) - \log\sum_{l_i=1}^{M} \exp AE(\mathbf {m}, \mathbf {l})\big\}
    \end{align}
where the expectation is taken over the joint distribution $p(\mathbf {m}, \mathbf {l})$. The following theorem presents the relationship between the estimated MI and the AE classifier.
\begin{theorem}
\label{theo: AE_MI_estimator}
Let $g_{\gamma}(\mathbf {m}, \mathbf {l}) = AE(\mathbf {m}, \mathbf {l})$. The infimum of the expected cross-entropy loss with output of Softmax activation function $S(AE(\mathbf {m}, \mathbf {l}))$ is equivalent to the MI between input and output variables $\mathbf {m}, \mathbf {l}$, respectively, up to constant $\log M$ under uniform label distribution.
\end{theorem}

\begin{proof} 
Let $g_{\gamma}(\mathbf {m}, \mathbf {l}) = AE(\mathbf {m}, \mathbf {l})$, then the lower bound is
\begin{align}
\mathbb{E}_{\mathbf {m}, \mathbf {l}} \log \frac{\exp AE(\mathbf {m}, \mathbf {l})}{\sum_{l_i=1}^{\mathbf{m}} \exp AE(\mathbf {m}, \mathbf {l})_{l_i}}.
\end{align}
If the distribution of the label is uniform, we can rewrite
\begin{align}
\mathbb{E}_{\mathbf {m}, \mathbf {l}} &\log \frac{\exp AE(\mathbf {m}, \mathbf {l})}{\frac{1}{M}\sum_{l_i=1}^{M} \exp AE(\mathbf {m}, \mathbf {l})_{l_i}} \notag\\ 
    &= \mathbb{E}_{\mathbf {m}, \mathbf {l}} \log \frac{\exp AE(\mathbf {m}, \mathbf {l})}{\sum_{l_i=1}^{M} \exp AE(\mathbf {m}, \mathbf {l})_{l_i}} + \log M \notag\\ 
    &= - S(AE(\mathbf {m}, \mathbf {l})) + \log M,
\end{align}
which is equivalent to the negative expected cross-entropy loss up to constant $ \text{log} M$. Hence, the infimum of the expected cross-entropy equals the MI between input and output variables. 
\end{proof}

Therefore, the secrecy optimization can be transformed into the cross-entropy optimization between the transmitted and received signals. Next, we construct a new cross-entropy loss function to optimize the MIMO-AE channel parameters.

\subsection{Security Loss Function}
The term security loss function was first proposed in~\cite{fritschek2019deep}, which is a mixed loss cross-entropy function aiming to maximize the ML between the transmit and receive symbols at the Rx and increase the cross-entropy between the transmitted and the received symbols at the eavesdropper. However, this assumption may not be available in practical scenarios because eavesdroppers are passive and would never give feedback for the training process. 
    In our work, we propose the following cross-entropy loss function:
    \begin{align}
L= (1 - \alpha)H(p_A(\mathbf{s}_k),p_B (\mathbf{s}_k) + \alpha H(p_A(\mathbf{w}_k),p_B(\mathbf{w}_k))
        \label{eq:cost}
    \end{align}
where $p_A(s_k)$ and $ p_B(s_k)$ are the probability mass functions of the information signals, and $ p_A(w_k)$ and $ p_B(w_k)$ are the resulting probability mass functions of the FJ signals at the Tx and the Rx, respectively, $H(\cdot)$ denotes the cross-entropy, and $\alpha$ is a parameter for the security and communication rate trade-off. Hence, minimizing $H(p_A(s_k),p_B (s_k))$ or maximizing the output probability of symbol $s_{ik}$ will decrease the output probability of all other symbols at the Rx. In contrast, maximizing $H(p_A(s_k),p_E(s_k))$ forces the system to reduce the output probability of the symbol $s_{ik}$, resulting in a higher probability on other symbols $s_{jk}$, $i \# j$.
In other words, this scheme can minimize the MI $I(FJ, B)$ and maximize $I(A, B)$. As discussed, the AE classification task with softmax and cross-entropy loss function is equivalent to the maximum MI between the input and output. Hence, minimizing the loss function~\eqref{eq:cost} will reduce the effect of FJ on the Rx and maximize $I(A, B)$. By contrast, the eavesdropper channel will be degraded with a high probability since the eavesdropper does not have any information about the FJ signals. 

The power allocation through the training process plays a key role in balancing security and BLER. The FJ generation layer generates the FJ signal $\mathbf{w}_k$ and concatenates with the transmit signal $\mathbf{s}_k$ to generate $\mathbf{x}_k$. The $\mathbf{x}_k$ is normalized in the batch normalization layer, which enforces power constraints. The training process employs unsupervised learning with the loss function $L_{bf}=-R_k^s$. , ensuring that the optimization of 
$R_k^s$ adheres to these constraints. Further, the parameter $\alpha$ in the security loss function $L$ is designed to minimize the effects of the jamming signal $\textbf{w}_k$ on the receiver, helping to reduce interference and maintain the SNR at Bob. While $\alpha$ is not a direct balancing factor between the secrecy rate and BLER, it increases the likelihood that the FJ signal will impact the eavesdropper. This adaptive mechanism is particularly important when only statistical CSI is available at the transmitter, requiring reliance on probabilistic estimates rather than precise channel conditions between the transmitter and eavesdropper. This approach is more realistic, as the transmitter cannot access the eavesdropper's channel. Moreover, we have observed that fine-tuning the parameter $\alpha$  leads to better convergence and helps avoid overfitting during autoencoder training compared to models that do not include this parameter.
\subsection{FJ Generation}
\begin{figure}[!tb]
    \centering
    % \hspace{-1.0cm}
    \includegraphics[width=1\linewidth]{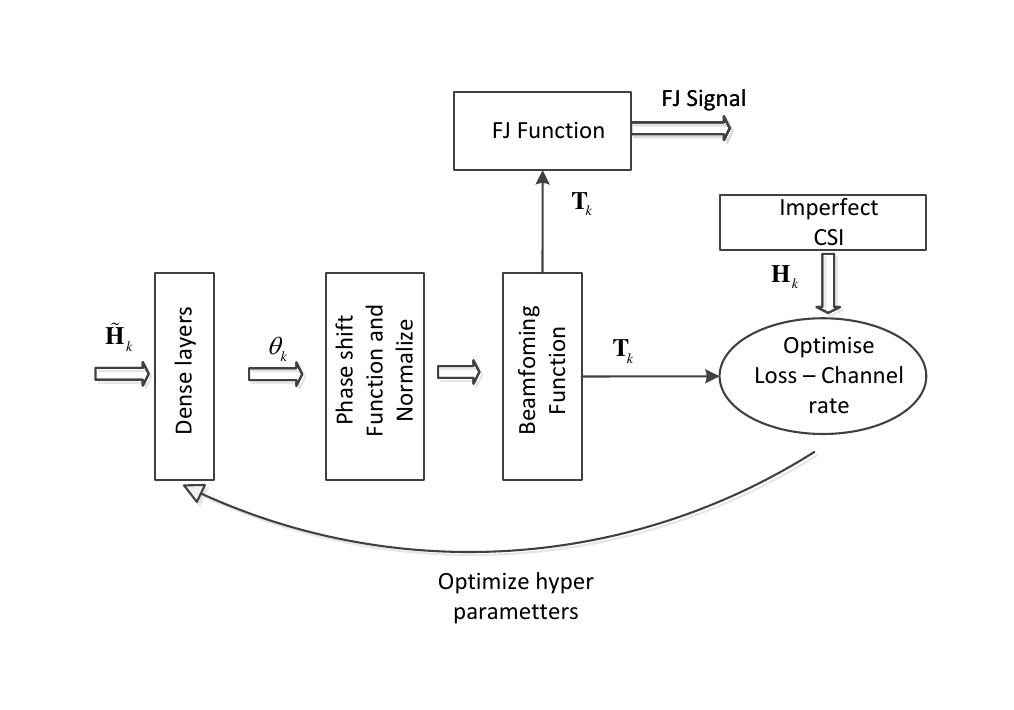}
    % \\[-.5cm]
    \centering
    \caption{Learning-based capacity driven FJ.}
    \label{fig:Beamformer}
\end{figure}

To generate FJ signals, we propose a learning-based capacity-driven FJ (LCD-FJ) model, illustrated in Fig.~\ref{fig:Beamformer}. The LCD-FJ structure consists of an embedding layer that converts complex values of the inputs into real values, fully connected layers, and three Lambda functions. The first Lambda function is used to encode the input to the azimuth. The second Lambda function is used to design the precoding matrix $\mathbf{T}_k$, and the data rate is then calculated. The training process aims to optimize the precoding matrix $\mathbf{T}_k$ by minimizing the loss function. As mentioned, the FJ signal is generated in the nullspace of $\mathbf{T}_k$, defined by the third Lambda function, so the FJ signal is as orthogonal as possible with the legitimate~channel.
The impact of FJ on eavesdropping suppression is significant. FJ degrades the eavesdropper’s SNR by strategically injecting noise into the transmission. In MIMO systems, beamforming techniques shape these noise signals to target the eavesdropper while preserving the quality of communication for the legitimate receiver. By projecting the noise onto the null space of the legitimate channel matrix, the noise minimally affects the legitimate receiver while significantly disrupting the eavesdropper’s reception. To address the challenges of imperfect CSI, our Learning-based Capacity-Driven FJ (LCD-FJ) model employs deep learning to optimize the FJ signal, ensuring it remains orthogonal to the legitimate channel. This minimizes interference and enhances communication security. Additionally, the AE-based FJ approach utilizes autoencoders to design FJ signals that maintain the quality of legitimate communication even in the presence of imperfect CSI. Finally, the MINE-based FJ approach, which functions without statistical CSI, maximizes the mutual information between the transmitter and receiver, providing robust security against eavesdropping by dynamically adapting to channel variations.

\section{Implementation}
Regarding the implementation details, the transmit messages, $\textbf{m}_k$, are discrete symbols representing data to be transmitted, generated uniformly and encoded into transmitted symbols $\mathbf{s}_k$. Since deep learning models operate on real numbers, a conversion occurs in the embedding layer, which separates the real and imaginary parts of complex symbols into real-valued vectors. The encoder includes three fully connected (dense) layers. The first layer expands the input symbols into a higher-dimensional space (e.g., expanding a vector of length 16 to 64) to capture more features. ReLU activation adds non-linearity. Batch normalization between layers stabilizes and accelerates training.

The final output of the encoder consists of the transmit symbols $\mathbf{x}_k$, which represent the encoded information symbols prepared for transmission over the MIMO system. These symbols include imaginary components. The re-parameterization process adapts these symbols for real-valued processing in the MIMO system. This process converts the complex symbols into a real-valued format by separating their real and imaginary parts and stacking them into a real-valued vector.
For instance, consider the original complex symbol $\mathbf{x}_k = \mathbf{s}_k + \mathbf{w}_k$, where $\mathbf{s}_k$ represents the information-bearing signal and $\mathbf{w}_k$ represents the Friendly Jamming (FJ) signal. The re-parameterized symbol, $\hat{\mathbf{x}}_k$, is formed by stacking the real and imaginary parts of $\mathbf{x}_k$ into a real-valued vector.
This conversion ensures that all complex parameters, such as $\mathbf{x}_k$, are transformed into real-valued components, making them suitable for real-valued operations within the MIMO system.

Decoding begins with the received signal $\mathbf{y}_k$, converted from complex to real values, mirroring the encoding process. The received signal is separated into real and imaginary parts and stacked as a real-valued vector for the decoder. The decoder, with three dense layers, processes the signal and reconstructs the original message. The first dense layer expands the input vector (e.g., from 64 to 256), enhancing the decoder's ability to interpret the signal. Batch normalization stabilizes training by normalizing inputs, reducing the risk of vanishing or exploding gradients, especially in noisy environments.

Subsequent dense layers reduce dimensionality back to the original message size. For example, after expanding to 256, layers reduce to 128 and then to the final output size, typically matching the original message length (e.g., 16). ReLU activation adds non-linearity to learn complex patterns. The decoder outputs the reconstructed message, aligning with the original transmitted data. In classification tasks, a Softmax function produces probabilities for each message symbol, while a linear function is used for regression.
The decoder minimizes reconstruction error, ensuring the decoded message closely matches the original. Then, the categorical cross-entropy loss function 
 is used to train the decoder to enhance both accuracy and reliability.

\section{MINE-based MIMO FJ}
In the previous section, we proposed the AEFJ method for secure communication, where the secrecy was optimized based on the E2E learning process with ImCSI. In this section, we propose a MINE-based FJ method to deal with the scenarios when even the SCSI is unavailable at the Tx.
\subsection{MINE-based MIMO Communication}
\label{subsection:MINE_Backgr}
The MINE concept proposed in~\cite{pmlr-v80-belghazi18a} estimates the MI between two random variables $X$ and $Y$ without knowing their distribution functions. In~\cite {fritschek2019deep}, the authors leveraged MINE for channel coding by estimating and maximizing the MI between the Tx and Rx symbols on the Gaussian single-input-singple-output (SISO) communication channel. In this work, we introduce the MINE-based MIMO communication as illustrated in  Fig.~\ref{fig:MINE_ENCODER}. The encoder and decoder resemble the AE-based communication. 
\begin{figure}[!tbh]
    \centerline{\includegraphics[width=1.0\linewidth]{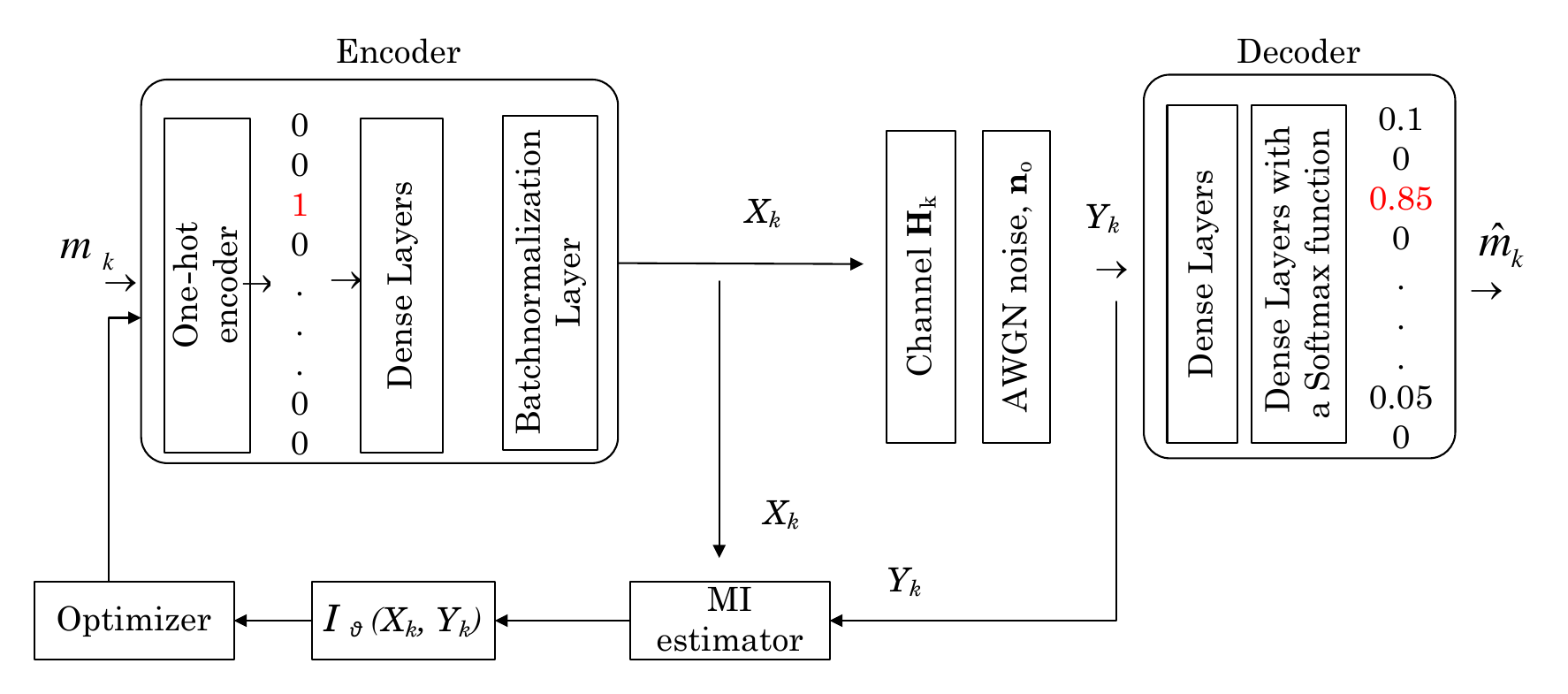}}
    \centering 
    \caption{The MINE-based MIMO channel.}
    \label{fig:MINE_ENCODER}
\end{figure} 

The MI between the channel input and output samples will be estimated and maximized by optimizing the weights of the encoder.
The MI between $X$ and $Y$ can be denoted as follows:
\begin{equation}
    I(X,Y)=D_{KL}(p_{XY})\parallel p_X \otimes p_Y),
\end{equation}  
where $D_{KL}$ is Kullback–Leibler divergence between the joint density $p_{XY}$ and the product of the marginal densities $p_X \otimes p_Y$. The Donsker-Varadhan representation \cite{pmlr-v80-belghazi18a} can then be applied to represent the KL divergence as
	\begin{equation}
		D_{KL}(p\parallel q) = \sup_{F_\Theta: \Omega  \to R} \mathbb{E}_p f - \log \mathbb{E}_q e^f,
		\label{eq:KL}
	\end{equation}
where the supremum is taken over all function classes $f$ such that the expectation is finite. Assume  $T_\theta (X_n, Y_n): X \times Y \to \mathbf{R}$ is a function in the class  function $f$, with parameters $\theta  \in \Theta $ in the family function $F_\Theta$. The transmit and receive symbols in one channel instance are denoted by $X_n$ and $Y_n$. When $p=p_{XY}$ and $q= {p_X} \otimes {p_Y}$ Then, $D_{KL}(p\parallel q)=I(X, Y)$ mutual information $I(X, Y)$ is given by
\begin{equation}
    I (X, Y) \ge \sup_{T_\theta \in F_\Theta} \mathbb{E}_p T_\theta (X_n, Y_n) - \log\mathbb{E}_q e^{T_\theta (X_n, Y_n)}.
\end{equation}
$I(X,Y)$ is estimated by its lower bound $I_\Theta(X,Y)$~\cite{pmlr-v80-belghazi18a}, known as the statistical MI between $X$ and $Y$ and is given by
\begin{equation}
	I_\Theta(X,Y)  =  \sup_{\theta  \in \Theta }\mathbb{E}_{p_{XY}} T_\theta (X_n, Y_n)  - \log\mathbb{E}_{{p_X} \otimes {p_Y}} e^{T_\theta (X_n, Y_n)}.
\label{eq:MI2}
\end{equation}

\begin{figure}[!tb]	
    \centering
    \includegraphics[width=.9\linewidth]{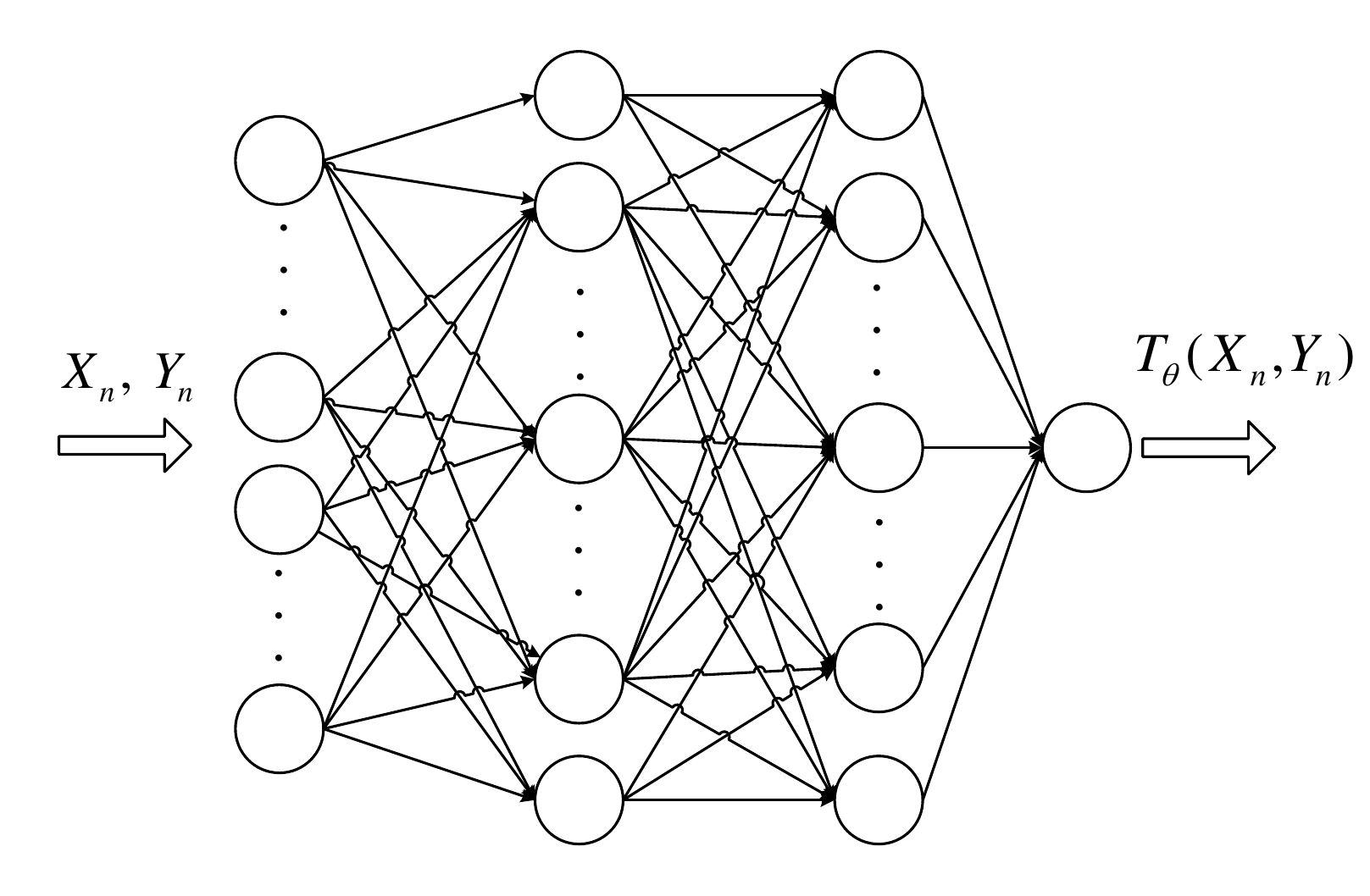}
    \caption{MI estimator. The inputs are the samples of Tx and Rx signals at the Rx and eavesdropper. The output is the estimated MI on the legitimate and illegitimate channels.}
    \label{fig:MINE}
\end{figure}
The estimator $T_\theta (X_n, Y_n)$ model contains two fully connected hidden layers with the ReLU activation function and an output node, which is depicted in Fig.~\ref{fig:MINE}.
The inputs are the samples of the Tx signals that follow the joint density $p_{XY}$ or $p_{XZ}$, and marginal densities $p_X$, $p_Y$ or $p_Z$, respectively. We then approximate the expectations by using the sample average method. The marginal $p_X$ and $p_Y$ can be derived by using the sampled average method, where the samples are shuffled from the joint distribution in the batch dimension~\cite {pmlr-v80-belghazi18a}. The estimated statistical MI between the Tx and the Rx is then can be expressed as follows:
\begin{align}
    \label {eq:SMI_lb}
    I_\Theta(X_n,Y_n) = \frac{1}{N} \sum_{i = 1}^N T_{\theta}(\mathbf{x}_i^n, \mathbf{y}_i^n)  - \log \frac{1}{N} \sum_{i = 1}^N e^{T_{\theta}(\mathbf{x}_i^n,\mathbf{\bar y}_i^n)},
\end{align}
where $\mathbf{x}_i^n$ and $\mathbf{\bar y}_i^n$ are the elements in the sets of samples $X_n$ and $Y_n$, respectively. The optimization process is performed alternatively between the MINE network and the encoder with the weights $\theta$ and $\phi$, respectively. This training process also optimizes the block error rate when the decoder is pre-trained separately from the encoder, making the training more flexible.

\subsection{MINE-based MIMO FJ Scheme}
\label{subsubsec:MINE-based FJ}

The main advantage of MINE is that it can estimate and maximize the MI between two random variables without the need for joint probability distribution between them. The goal here is to maximize the MI between the Tx and the Rx, $I(A, B)$, while minimizing the MI between the Tx and the eavesdropper, $I(A, E)$. The MINE-based FJ is presented in Algorithm~\ref{MINE_FJ_Algo}. We assume that the eavesdropper has the full knowledge of CSI. 
\begin{algorithm}[!tb]
\caption{MINE-based MIMO Friendly Jamming}
\label{MINE_FJ_Algo}
\begin{algorithmic}[1]
    \State At time $k$: Generate $\mathbf{w}_k=\mathbf{Z}_k\mathbf{v}_k$, orthogonal to $\mathbf {H}_k^\dag$
    \State Tx transmits $\mathbf {x}_k=\mathbf {s}_k  + \mathbf {w}_k$
    \State Rx receives $\mathbf {y}_k$
    \For{\textit{i=1 to Iteration}}
    \State Maximize estimated MI $I_{AB}^{(i)}$
    \State Save estimated $I_{AB}^{(i)}$
    \State Load $I_{AB}^{(i)}$ as loss value to optimize encoder weights by performing gradient descent steps
    \State Set $i=i+1$
        \If {$I_{AB}^{(i+1)} \leq I_{AB}^{(i)}$}
        \State break
        \EndIf
    \EndFor		
\end{algorithmic}
\end{algorithm}
The receive symbol in channel instance $n$ is denoted as $Z_n$. Following \eqref{eq:SMI_lb}, the estimated MI between $X_n$ and $Z_n$ can be expressed as
\begin{align}
    I_\Theta(X_n,Z_n) = \frac{1}{N}\sum_{i = 1}^N T_{\theta_2}(\mathbf{x}_i^n, \mathbf{ z}_i^n)  - \log \frac{1}{N}\sum_{i = 1}^N e^{T_{\theta_2} (\mathbf{x}_i^n, \mathbf{\bar z}_i^n)},
\end{align}
where $\theta_2$ is the parameters of the neural estimator at the eavesdropper.
The encoder and MI estimator structures, i.e., $I_\Theta$ network, remain unchanged as described in Section~\ref{subsection:MINE_Backgr}. The MINE-based FJ scheme will combine the MI estimator and encoder with the injected FJ signals, as illustrated in Fig.~\ref{fig:MINE-FJ}. The model takes the input as the sampled signals $X_n$ and $Y_n$ to estimate $I_\Theta (A, B)$. The value of $I_\Theta (A, B)$ is fed back to the encoder for adjusting the hyperparameters, weights, and biases. The same process is repeated ultill $I_\Theta (A, B)$ is~maximized.
    %The main advantage of this design is that the learning works with the channel density distribution and estimates a function of the channel.
\begin{figure}[!tb]
    \centerline{\includegraphics[width=\linewidth]{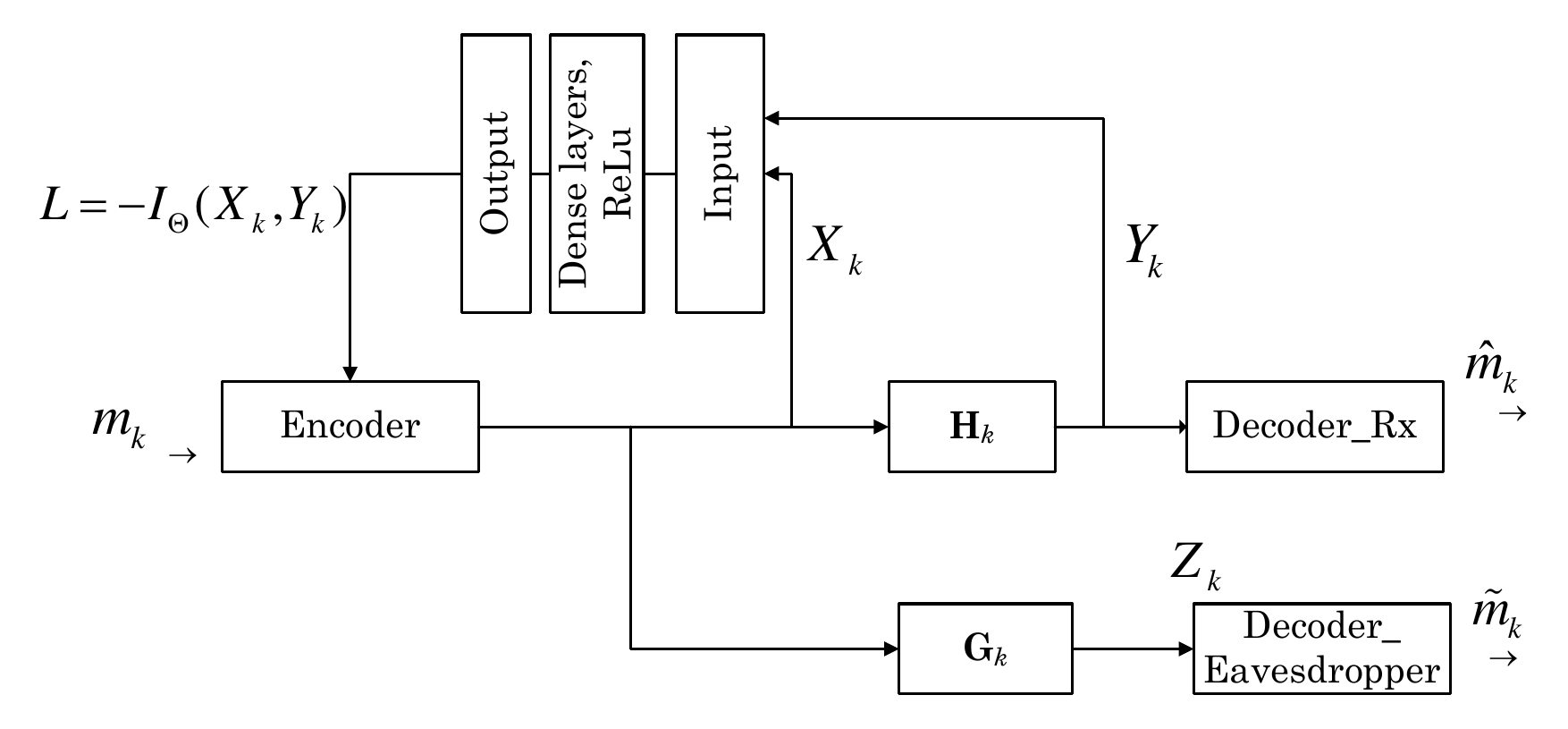}}
    \centering
    \caption{MINE-based FJ.}
    \label{fig:MINE-FJ}
\end{figure}    
% To achieve the security requirement and reliable transmission, we  propose a new security loss function based on MI with the control coefficient $\beta$ as follows:

To fulfill both security and transmission reliability objectives, we introduce the following novel security loss function:
    \begin{equation}
        L_\text{MINE} = \beta I_\Theta(X_n,Y_n) + (1 - \beta )I_\Theta(X_n,Z_n),
    \end{equation}
where $\beta$ is a coefficient that controls the trade-off between communication rate and security. Then, the loss function can be further written as follows:
\begin{align}
    L_\text{MINE} &= \frac{\beta }{N}\ \sum \limits_{i = 1}^N T_\theta({\bf{x}}_i^n,{\bf{y}}_i^n) - \beta \log \frac{1}{N} \sum \limits_{i = 1}^N e^{{T_\theta}({\bf{x}}_i^n,{\bf{\bar y}}_i^n)} \notag
    \intertext{}
    - \frac{{1 - \beta}}{N} & \sum \limits_{i = 1}^N T_\theta ({\bf{x}}_i^n,{\rm{\bar z}}_i^n) + (1 - \beta )\log \frac{1}{N} \sum \limits_{i = 1}^N e^{{T_\theta }({\bf{x}}_i^n,{\rm{\bar z}}_i^n)}.
        \label{eq:Loss_mine}
\end{align}
% The positive constant $\beta$ represents the trade-off between the communication and secrecy rates. 

Furthermore, we consider the worst case when the noise power at the eavesdropper $\sigma_e^2=0$  \cite{goel2008guaranteeing}. Then, the covariance of noise at Eve in \eqref{eq:Cov_noise_Eve} becomes $\mathbf{K}_k'=(\mathbf {G}_k\mathbf {Z}_k\mathbf {Z}_k^\dag \mathbf{G}_k^\dag)\sigma_v^2$. To simplify the notation, all occurrences of $R_k ^s$ will be understood as representing GSC from this point forward. Then, the guaranteed secrecy rate (GSC) is given by
\begin{align}
    R_k ^s &= {I_\Theta }({X_n},{Y_n}) - R_{AEmax} \notag\\
        &= {I_\Theta }({X_n},{Y_n}) -\log 
        \frac {\det (\mathbf {K'}_k+\mathbf {G}_k^\dag\mathbf {Q}_s\mathbf {G}_k)} {\det\mathbf {K'}_k}. \label{eq:Cs_gsc}
\end{align}
That optimization problem can be resolved by performing SVD on $\mathbf{H}_k$ and the second-order perturbation analysis~\cite{mukherjee2010robust}. However, this method requires exponential complexity when the number of antennas increases. Thus, we leverage the non-convex optimization capability provided by deep neural networks and then directly solve the problem with sufficient training data by considering the CSI error as an input for the training process.
\section{Simulation Results and Discussion}
\label{sec:sim}
\subsection{Simulation Setups and Parameters}

In this section, we use the Montecarlo method to evaluate the effectiveness of the proposed approaches. The secrecy rate is averaged over $10^6$ iterations. The channels  $\mathbf{H}_k $ and $ \mathbf{G}_k $ are generated as complex Gaussian matrix~\cite{telatar1999capacity}, where $h _{i,j}$ and $ g _{i,j}$ are assumed to be i.i.d. Gaussian with $ \mathbb{E} |h|_{i,j}^2 = \mathbb{E}|g| _{i,j}^2 = 1 $. The real and imaginary parts of the entries in the channel matrices $\mathbf{H}_k$ and $\mathbf{H}_k$ are standard Gaussian i.i.d.. We generate a complex-valued channel matrix where the real and imaginary parts are normally distributed. This results in a Rayleigh-distributed channel matrix magnitude commonly used to model fading in wireless communication channels. The phase of the channel follows a uniform distribution from $0$ to $2\pi$. The total power $ P_{\max} $ is normalized by the power of AWGN noise $ \sigma_n^2= \sigma_e^2=1 $. The SNR values in the training are set to vary from $10$~dB to $25$~dB, guaranteeing the generality of the trained model. For analysis, the eavesdropper is assumed to have the same neural architecture decoder as the Rx, as in Fig.~\ref{fig:AEFJ}. The AE at the Tx-Rx channel inducing FJ and the network at the eavesdropper are trained simultaneously as a one-input-two-output network with the loss function in~\eqref{eq:cost}.

To see the advantages of the proposed method, we compare our method with two baselines proposed in \cite{goel2008guaranteeing} and \cite{mukherjee2010robust}. The former first uses multiplexing as the precoding technique and then adopts executive search to optimize the average secrecy rate $R_{kex}^s$. Differently, the latter is based on the perturbation/error second-order analysis, called ``So'', to optimize the average secrecy rate $R_{kSo}^s$. Further, we use the state-of-the-art DL library Tensor Flow with Adam Optimizer for training. The input and output layers have $16$ neurons, representing a symbol of $4$ bits. The channel layer includes $N_t$ neuron width representing the number of transmit antennas. The network architecture and its parameters are summarized in Table~\ref{tab:Parameter_setting}.

\begin{table}[!tb]
    \caption{Architecture and Parameters of LCD-FJ Model}
    \label{tab:Parameter_setting}
    \centering
    \begin{tabular}{p{0.4\linewidth} p{0.25\linewidth} p{0.2\linewidth}}
    \hline\hline
    Layer (type), number of nodes & Output shape & \# parameters \\
    \hline
    imperfect\_CSI (InputLayer) & [(None, 1, 2, 64)] & 0 \\
    batch\_normalization\_24 (BatchNormalization) & (None, 1, 2, 64) & 256 \\
    flatten\_8 (Flatten) & (None, 128) & 0 \\
    batch\_normalization\_25 (BatchNormalization) & (None, 128) & 512 \\
    dense\_24 (Dense) & (None, 256) & 33,024 \\
    batch\_normalization\_26 (BatchNormalization) & (None, 256) & 1,024 \\
    dense\_25 (Dense) & (None, 128) & 32,896 \\
    dense\_26 (Dense) & (None, 64) & 8,256 \\
    perfect\_CSI (InputLayer) & [(None, 64)] & 0 \\
    lambda\_16 (Lambda) & (None, 64) & 0 \\
    SNR\_input (InputLayer) & [(None, 1)] & 0 \\
    lambda\_17 (Lambda) & (None, 1) & 0 \\
    \hline
    \end{tabular}
\end{table}

\subsection{Experimental Results}
We first implement the AE-based FJ model in the case of perfect CSI and compare its performance with that of the exhaustive search method in~\cite{goel2008guaranteeing}, as shown in Fig.~\ref{fig:Full_CSI_COMPARE}. In the AEFJ method, the channel matrix is considered an extra input parameter at the encoder and decoder in the AEFJ model. The parameter $\alpha$ was chosen as $0.5$ to make the balance between transmission and capacity. We consider two different values of the number of antennas at the Tx; $N_t = 10, 20$. As can be seen, the average secrecy rate of the proposed AEFJ method is close to that of the optimal scheme, i.e., the exhaustive search~\cite{goel2008guaranteeing}), when the number of antennas at the Tx is small, $N_t = 4$, confirming the effectiveness of the proposed AEFJ scheme. This can be explained by the capability of the AE to learn the features of channel distribution with sufficient data and perfect CSI.

\begin{figure}[!tb]    
 \center
 \includegraphics[width=\linewidth]{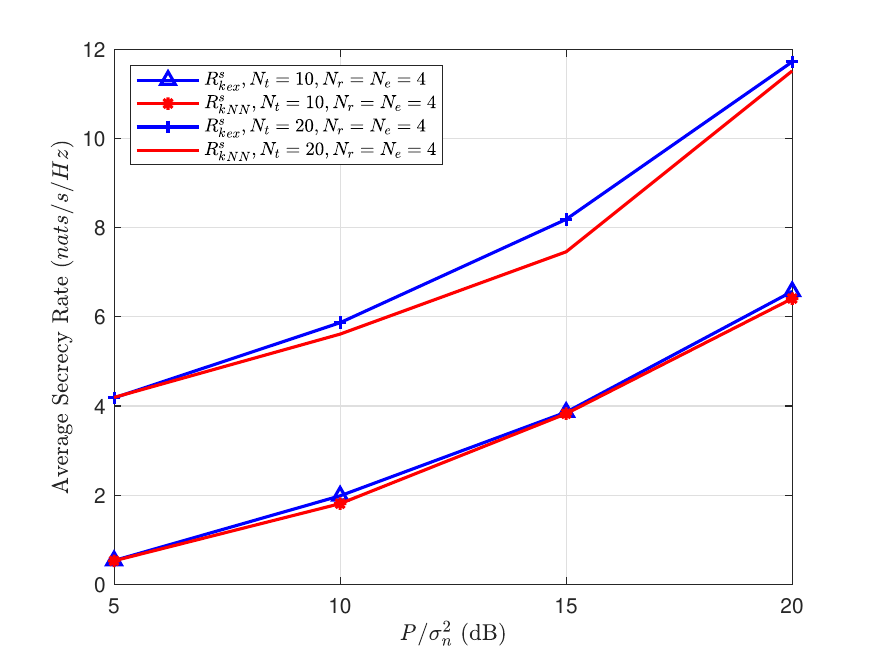}
    \caption{Average secrecy rate versus $P_{\max}/\sigma_n^2$ under ImCSI.}
    \label{fig:Full_CSI_COMPARE}
\end{figure}

Regarding BLER, Fig.~\ref{fig:BLER_FullCSI} shows the BLER of the proposed method at the Rx and the eavesdropper. By assuming that the eavesdropper and the Rx have the same model, we can see that the BLER at the eavesdropper's decoder is much higher than the Rx's. Although the BLER of the proposed method is higher than that of the ML-based method, the former provides much less computational complexity than the latter. 

\begin{figure}[!tb]
    \centerline{\includegraphics[width=\linewidth]{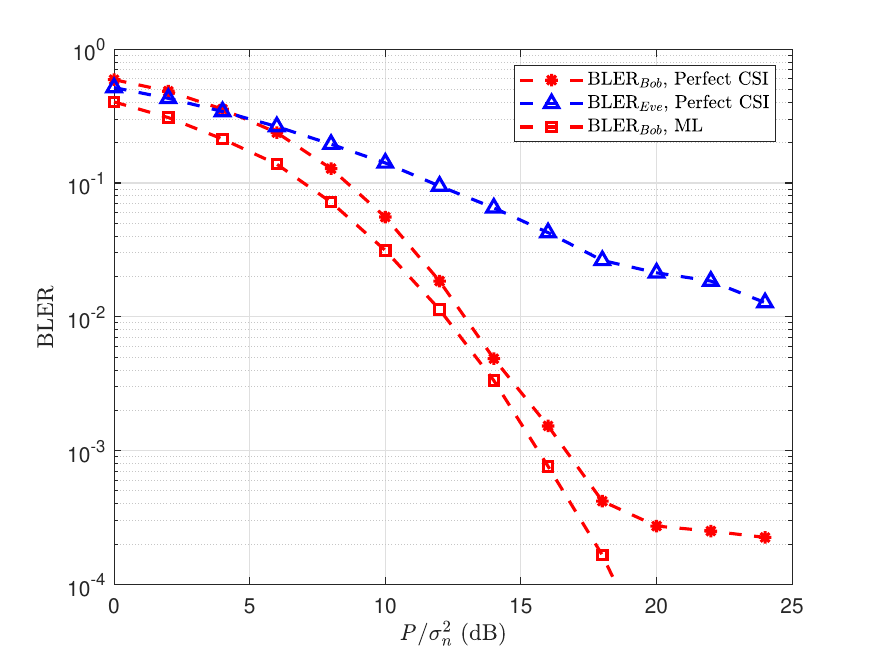}}
    \caption{BLER at Rx and eavesdropper with FJ; $\mathbf{H}_k$ perfectly known.}
    \label{fig:BLER_FullCSI}
\end{figure}

\begin{figure}[!tb]
    \centering
    \centerline{\includegraphics[width=\linewidth]{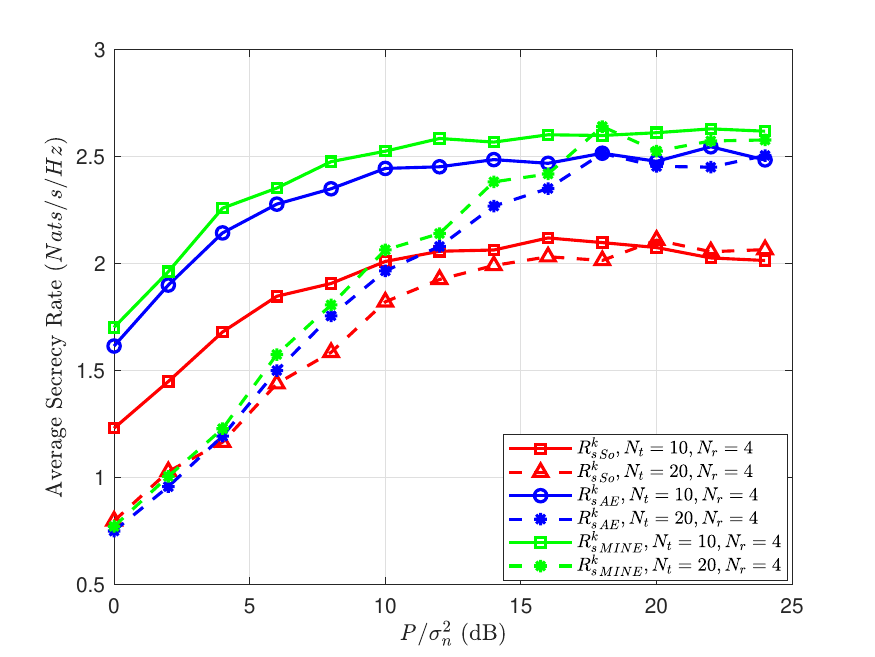}}
    \caption{Average secrecy rate versus $P/ {\sigma_n}^2$ (dB).}
    \label{fig:CS_SNR}	
\end{figure} 	    

Regarding ImCSI, our target is to compare the average secrecy rate achieved by our method with the work \cite{mukherjee2010robust}, where full CSI is unavailable at the Tx. The parameter $\sigma_e=-20 $~dB or $0.1$ represents the level of perturbation or channel estimation error. The security communication model is illustrated in Table~\ref{tab:quasi_static2}.
\begin{table}[!tb]
    \caption{Architechture and Parameters of AEFJ with SCSI}
    \label{tab:quasi_static2}
    \centering
    \begin{tabular}{m{2.7cm} m{5.3cm}}
    \hline\hline
    Encoder & Encode a message $m_k$ to $\mathbf{s}_k$ \\
    \hline
    LCD FJ & $\mathbf{w}_k$ orthogonal to $\mathbf{s}_k$ \\
    \hline
    Power constraints/norm & Normalize average power \\
    \hline
    Channel & Generate random complex channel $\mathbf{H}_k$ \\
    \hline
    Lambda & Matrix multiplication $\mathbf{x}_k$ with $\mathbf{H}_k$ \\
    \hline
    Dense & Simulate Tx-Rx channel, estimating mapping function $Q(\mathbf{x}_k, \tilde {\mathbf{H}_k)}$ \\
    \hline
    Softmax function & Calculate loss to optimize MI \\
    \hline
    \end{tabular}
\end{table}
As shown in Fig.~\ref{fig:CS_SNR}, the proposed AEFJ scheme offered a better average secrecy rate than So based-method~\cite{mukherjee2010robust}. We note that the secrecy rate increased linearly at the low SNR region but tended to be saturated at high normalized transmit power, i.e., from $10$ dB to $25$ dB. In addition, the higher secrecy rate can be seen in the MINE-based method compared to AEFJ~\cite{10075335}. The reason is the former optimizes the MI directly while the latter depends on the cross entropy loss. Moreover, the higher the number of transmit antennas, the higher the secrecy rate for both schemes. This stems from the fact that the proposed AE architecture can effectively perform denoising and capturing features of the channel with sufficient data.
% \subsection{MINE-based FJ}
The MINE security loss function \eqref{eq:Loss_mine} is used to evaluate the performance of the security approach.  
The network architecture of the MINE-FJ scheme is shown in Table~\ref{tab:ImCSI}.
\begin{table}[!tb]
    \caption{Description of function layers in AEFJ with ImCSI}
    \label{tab:ImCSI}
    \centering
    \begin{tabular}{m{2.7cm} m{5.3cm}}
    \hline\hline
    Input & Concatenate $x_k$ and $\tilde{H}_k$, converted to real domain from complex domain \\
    \hline
    Power constraints/norm & Normalize average power \\
    \hline
    Channel & Generate imperfect complex channel $\tilde{H}_k$ \\
    \hline
    Lambda & Matrix multiplication $x_k$ with $\tilde{H}_k$ \\
    \hline
    Hidden & Simulate Tx-Rx channel, estimating mapping function: $Q(x_k, \tilde{H}_k)$ \\
    \hline
    Output Layer & $\hat{x}_k$, and the activation function is soft-max for a reconstruction problem \\
    \hline
    Optimizer & Adam optimizer \\
    \hline
    Loss Function & Equation (18) \\
    \hline
    \end{tabular}
\end{table}

Fig.~\ref{fig:AveC} illustrates the secrecy rate with two different values of $\beta$, where the number of transmit antennas is $N_t=3$, each with $400$ iterations and the batch size of $20,000$. We observe that the higher the value of $\beta$, the higher the secrecy rate was obtained. This shows the trade-off between the communication and secrecy rates due to the influence of the FJ signal. Fig.~\ref {fig:BLER} shows BLER at both the Rx and the eavesdropper in a range of SNR values before and after a secure communication was applied by AEFJ. The black curves represent BLER at Rx when using MINE-based FJ with and without FJ. The red curves represent BLER at Rx when using AE-based FJ with and without FJ, whereas the blue curve illustrates the BLER at the eavesdropper. We notice two important points. Firstly, the BLER at Rx increased, with the gap of 3dB at $10^{-4}$ of BLER. It can be explained that a portion of power is allocated for jamming. Secondly, there was a significant increase in the BLER at the eavesdropper using AEFJ. The higher BLER at the eavesdropper is because of the effect of FJ signals on the eavesdropper channel, which makes the information signals undecodable at the eavesdropper.

\begin{figure}[!tb]
    \centerline{\includegraphics[width=\linewidth]{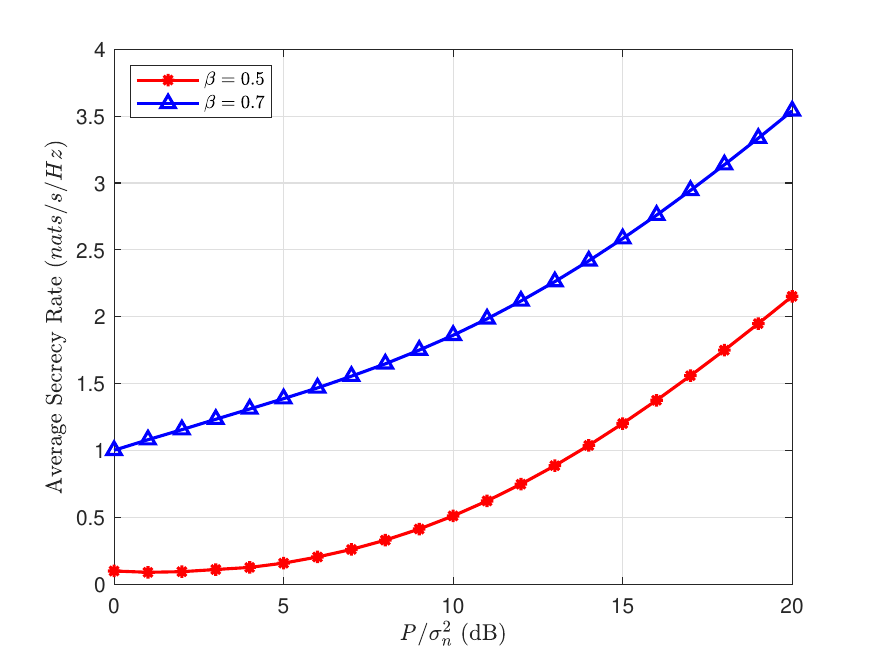}}
    \caption{Average secrecy rate with different values of $\beta$.}
    \label{fig:AveC}
\end{figure}

\begin{figure}[!tb]
    \centerline{\includegraphics[width=\linewidth]{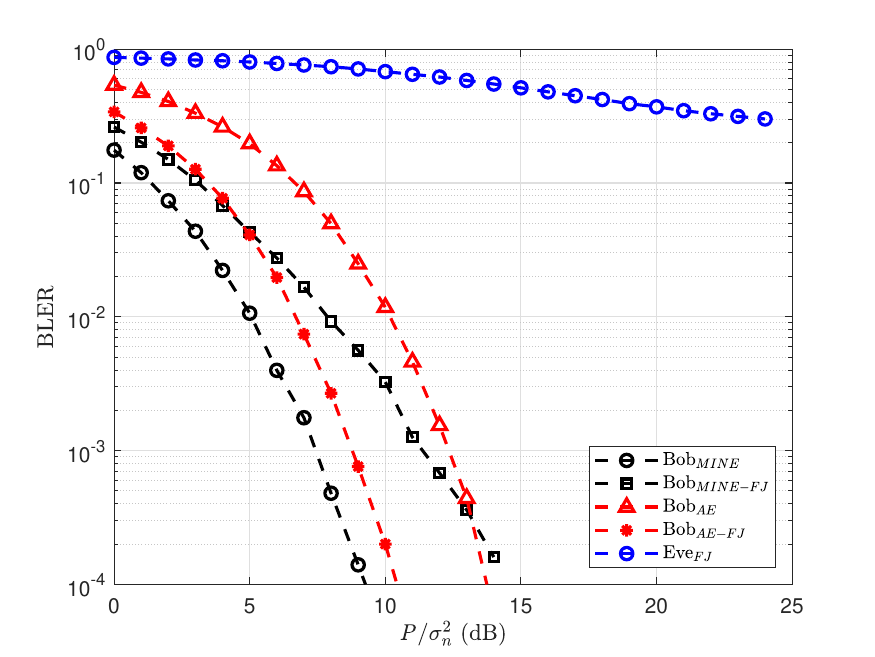}}
    \caption{BLER at Rx and eavesdropper with AEFJ and MINE-FJ.}
    \label{fig:BLER}
\end{figure}

\begin{figure}[!tb]
    \centering\centerline{\includegraphics[width=\linewidth]{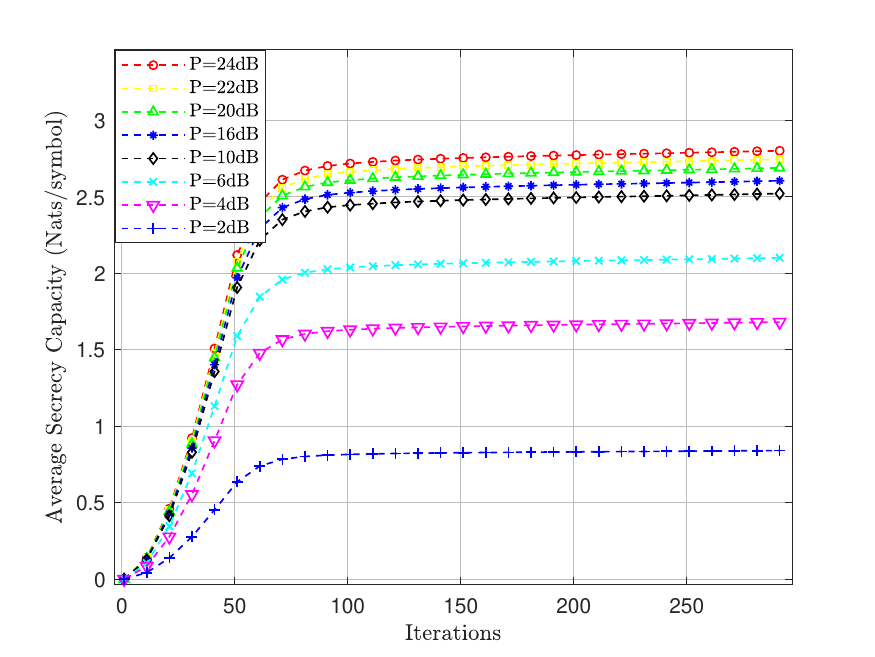}}
    \caption{Secrecy rate vs. number of iterations; $N_t=10$, $N_r=4$ and $N_e=4$.}
    \label{fig:MINE-MI}
\end{figure}

In Fig.~\ref{fig:MINE-MI}, the convergence of MINE-FJ is provided with different SNR levels using the security model in Fig.~\ref{fig:MINE-FJ}. It can be seen that the MINE approach quickly converged to a stationary point after about  $100$ iterations. In addition, the secrecy rate improved dramatically when the SNR increased but tended to saturate for a sufficiently high SNR. This saturation can be explained as when the distribution of the transmit is close to receiving samples.

    % \begin{figure}[h]
    %     \centerline{\includegraphics[width=.9\linewidth]{./figures/BLER_MINE_2}}
    %     \caption{BLER at Rx and eavesdropper using MINE security.}
    %     \label{fig:BLER-MINE}
    % \end{figure}
    
The relationship between the average secrecy rate and the number of transmit antennas is provided in Fig.~\ref{fig:R_S_vs_NT}. It shows that the average secrecy rate increased rapidly with the number of antennas $N_t$. This is attributed to the fact that the higher the number of transmit antennas at the Tx, the more the degree of freedom can be added to the system to transmit the desired signal and design the FJ signal more effectively.
\begin{figure}[!tb]
    \centering                   \centerline{\includegraphics[width=\linewidth]{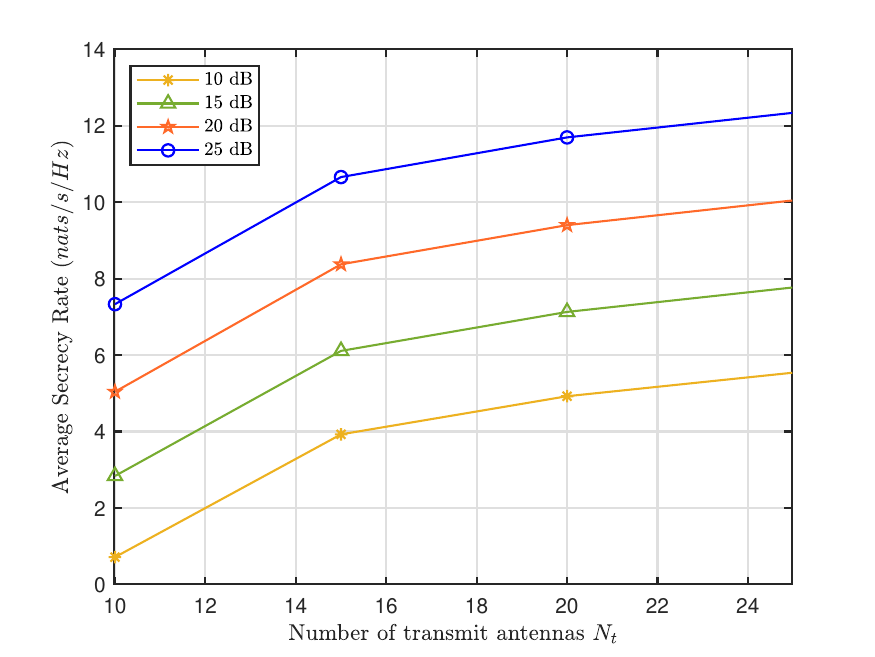}}
    \caption{Secrecy rates vs. $N_t$.}
    \label{fig:R_S_vs_NT}
\end{figure}
In Fig.~\ref{fig:MI-BER}, we show a trade-off between the secrecy rate and BLER with and without using FJ. As expected, the higher the BLER, the lower the average secrecy rate is obtained. 
% In other words, the decoding error at Rx decreases along with the secrecy rate and vice versa, which is aligned well with the previous discussion in~\ref{subsubsection:Rethinking}. 
\begin{figure}[!tb]
    \centerline{\includegraphics[width=\linewidth]{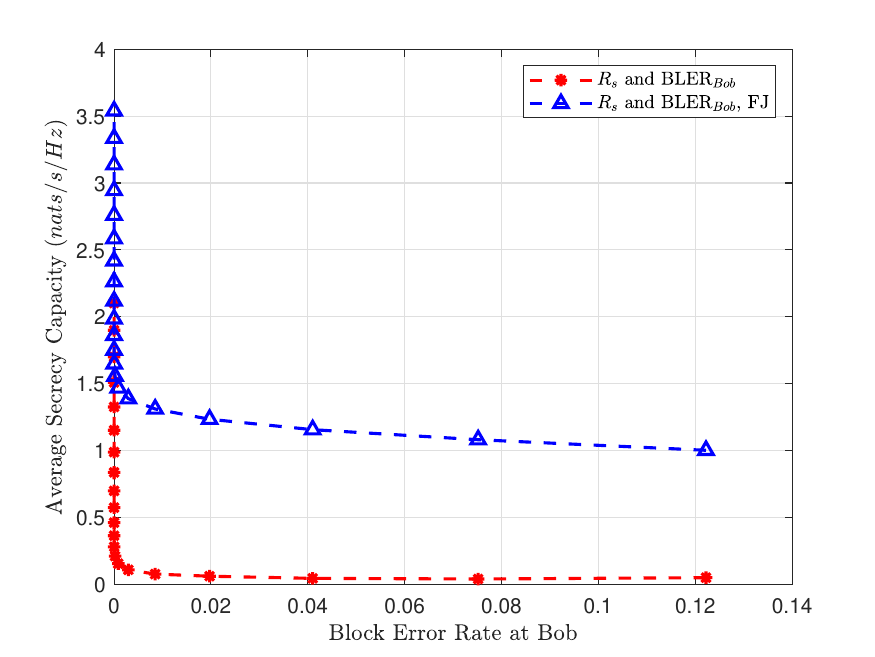}}
    \caption{Average secrecy rate and BLER at Rx; $N_t=10$, $N_r=4$ and $N_e=4$.}
    \label{fig:MI-BER}
\end{figure}

The loss function of the AE-based FJ model is designed to optimize secrecy and communication reliability, and it is minimized during training. Fig.~\ref{fig:training_loss_AEFJ_1} illustrates a consistent decrease in loss over the training epochs, eventually stabilizing as the model parameters converge to an optimal solution. This convergence highlights the model's ability to effectively implement the friendly jamming technique under imperfect CSI conditions, maximizing the secrecy rate while minimizing the block error rate. The figure presents both training and validation losses over 8 epochs. The training loss decreases from 4.8 to 4.0, while the validation loss shows a slight reduction, indicating steady improvement. Most loss reduction occurs in the first 2 epochs, with losses dropping by about 0.4. After this initial phase, convergence slows, leading to more minor changes per epoch. This pattern reflects efficient learning and good generalization, as both losses follow similar trends and converge steadily.
% The convergence speed is 
\begin{figure}[h]
\centering
\includegraphics[width=0.85\linewidth]{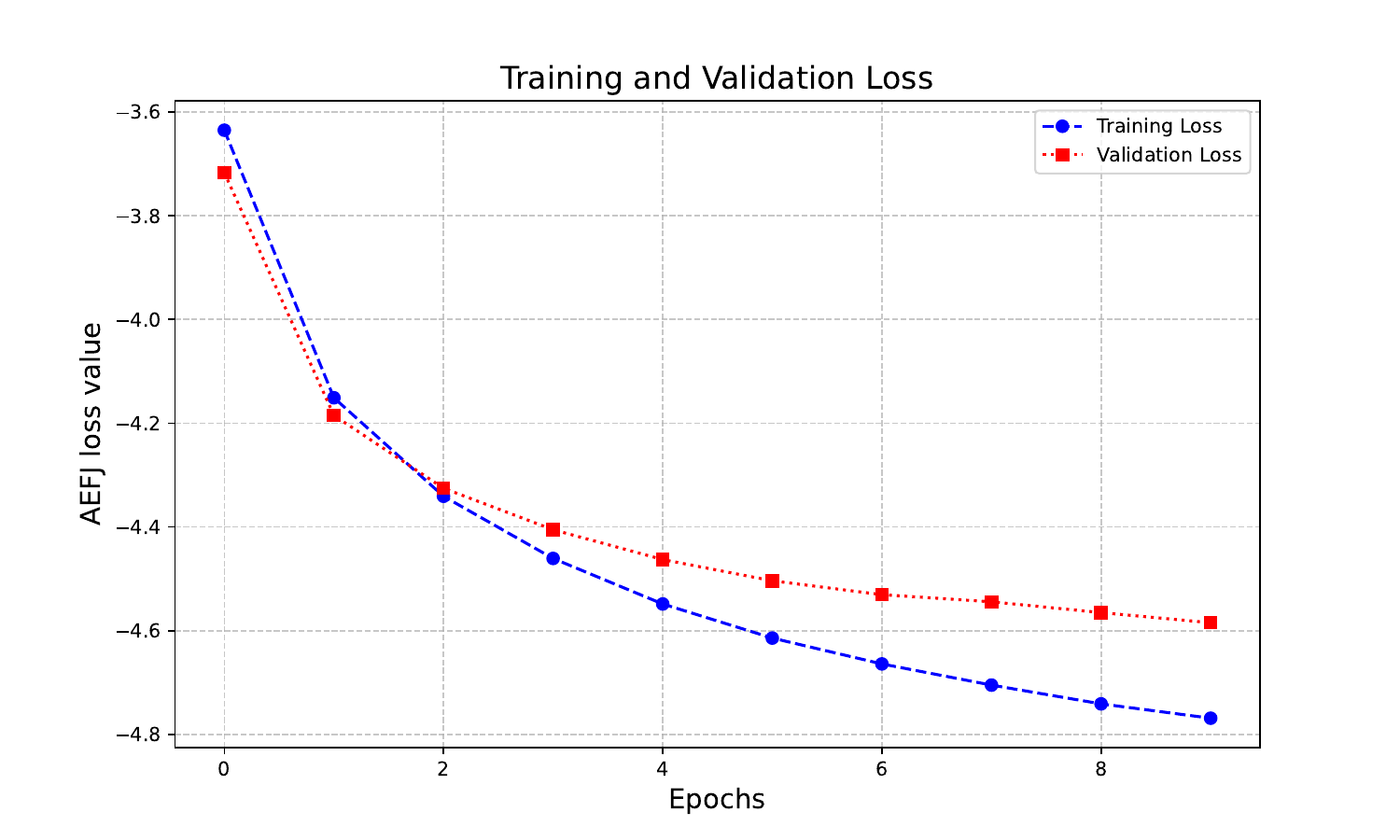}
\caption{AE-based FJ training Loss.}
\label{fig:training_loss_AEFJ_1}
\end{figure}

\begin{figure}[t]
\centering
\includegraphics[width=0.85\linewidth]{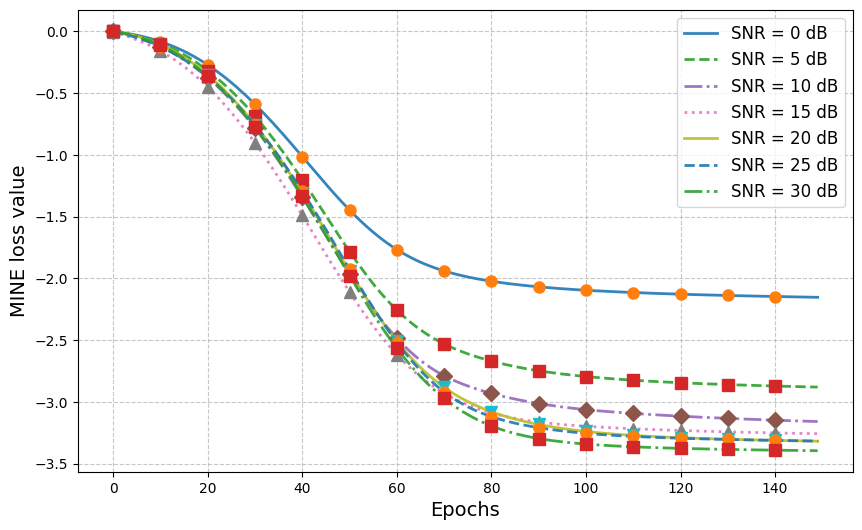}
\caption{MINE-based FJ training Loss.}
\label{fig:MINE_100}
\end{figure}

The training graph for the MINE-based FJ model in Fig. \ref{fig:MINE_100}. This illustrates the convergence of mutual information, represented as the negative of the loss during training. The model optimizes secrecy without requiring CSI at the transmitter and shows a distinct convergence pattern. The plot illustrates the loss behaviour across epochs for various SNR values (0 dB to 30 dB), revealing that higher SNRs lead to faster convergence and lower final loss values. At lower SNRs, such as 0 dB and 5 dB, the loss decreases slowly and stabilizes at higher levels, indicating poorer performance. In contrast, for higher SNRs (20 dB, 25 dB, and 30 dB), the loss decreases rapidly and stabilizes after around 100 epochs, reflecting improved efficiency in minimizing mutual information loss. These trends demonstrate that higher SNRs facilitate more effective communication by reducing noise, allowing for better signal separation and lower loss values.
        
\subsection{Complexity}

This subsection compares the computational complexity as the number of floating-point operations (FLOPs) between AEFJ and the conventional methods. Herein, we only counted the complexity in the deployment stage since the training stage can be seen as offline. The number of FLOPs of a dense layer is referred to in the work in~\cite{9134818}, which is equal to $(2N_{Input}-1)N_{Output}$, where $N_{Input}$ and $N_{Output}$ are the input and output dimensions, respectively. For example, when $N_t = 20$ and the transmit symbol is encoded into a one-hot vector and considering the architecture of our model in Table~\ref{Tab:complexity}, the FLOPs is $6,144$. For the exhausting search and ``So" method mentioned previously, the asymptotic computational complexity is in the order of $\mathcal{O} (N_t^3)$ as they perform SVD and matrix inversion.

 \begin{table}[!tb]
    \caption{Parameters and FLOP of AEFJ.}
    \centering
    \begin{tabularx}{0.49\textwidth} { 
            | >{\raggedright\arraybackslash}X 
            | >{\centering\arraybackslash}X 
            | >{\centering\arraybackslash}X | }
        \hline
        Name of Layers &Input and Output Dimensions% (number of input) 
            & FLOPs \\
        \hline\hline
        Input Layer  & 16; 64 & 2,048 \\
        \hline
        Dense Layer 2  & 64; 8  & 1,024  \\
        \hline
        Dense Layer 3  & 8; 64  & 1,024  \\
        \hline
        Output layer  & 64; 16  & 2,048  \\
        \hline
    \end{tabularx}
    \label{Tab:complexity}
\end{table}

\section{Conclusion}\label{sec:VII}
In this paper, we have introduced a novel DL-based FJ approach to deal with the eavesdropping issues in MIMO-based systems. First, we have proposed the AEFJ scheme by leveraging the Auto Encoder-based E2E learning at both the Tx and the Rx, which has shown to achieve both communication secrecy and reliability compared to the conventional methods. In addition, we have leveraged MINE to design a robust security scheme that can deal with different levels of availability of CSI at Tx, from perfect to statistical or even unknown CSI. Simulation results have showed the comparable security performance of MINE to AEFJ under the cross-entropy security loss function. Further, the secrecy rate was optimized independently at the Tx and the Rx using the MI neural estimator. Thus, the MINE-based FJ is promising for applications that require fast deployment and lightweight but effective security methods. Our approach here can be extended to other practical settings, e.g., for IRS/RIS-aided systems \cite{9852985, 10437125825}.
 % The details of the such as IoT devices.
 
 \bibliographystyle{IEEEtran}
	\bibliography{ref_Journal}

% Generated by IEEEtran.bst, version: 1.14 (2015/08/26)
\begin{thebibliography}{10}
\providecommand{\url}[1]{#1}
\csname url@samestyle\endcsname
\providecommand{\newblock}{\relax}
\providecommand{\bibinfo}[2]{#2}
\providecommand{\BIBentrySTDinterwordspacing}{\spaceskip=0pt\relax}
\providecommand{\BIBentryALTinterwordstretchfactor}{4}
\providecommand{\BIBentryALTinterwordspacing}{\spaceskip=\fontdimen2\font plus
\BIBentryALTinterwordstretchfactor\fontdimen3\font minus \fontdimen4\font\relax}
\providecommand{\BIBforeignlanguage}[2]{{%
\expandafter\ifx\csname l@#1\endcsname\relax
\typeout{** WARNING: IEEEtran.bst: No hyphenation pattern has been}%
\typeout{** loaded for the language `#1'. Using the pattern for}%
\typeout{** the default language instead.}%
\else
\language=\csname l@#1\endcsname
\fi
#2}}
\providecommand{\BIBdecl}{\relax}
\BIBdecl

\bibitem{mukherjee2014principles}
A.~Mukherjee, S.~A.~A. Fakoorian, J.~Huang, and A.~L. Swindlehurst, ``Principles of physical layer security in multiuser wireless networks: A survey,'' \emph{IEEE Commun. Surv. Tutor}, vol.~16, no.~3, pp. 1550--1573, 2014.

\bibitem{zhang2015secure}
Y.~Zhang, Y.~Shen, H.~Wang, J.~Yong, and X.~Jiang, ``On secure wireless communications for iot under eavesdropper collusion,'' \emph{IEEE Transactions on Automation Science and Engineering}, vol.~13, no.~3, pp. 1281--1293, 2015.

\bibitem{chu2023countering}
N.~H. Chu, N.~Van~Huynh, D.~N. Nguyen, D.~T. Hoang, S.~Gong, T.~Shu, E.~Dutkiewicz, and K.~T. Phan, ``Countering eavesdroppers with meta-learning-based cooperative ambient backscatter communications,'' \emph{arXiv preprint arXiv:2308.02242}, 2023.

\bibitem{bernstein2017post}
D.~J. Bernstein and T.~Lange, ``Post-quantum cryptography,'' \emph{Nature}, vol. 549, no. 7671, pp. 188--194, 2017.

\bibitem{bloch2008wireless}
M.~Bloch, J.~Barros, M.~R. Rodrigues, and S.~W. McLaughlin, ``Wireless information-theoretic security,'' \emph{IEEE Trans. Inf. Theory}, vol.~54, no.~6, pp. 2515--2534, 2008.

\bibitem{shannon1949communication}
C.~E. Shannon, ``Communication theory of secrecy systems,'' \emph{The Bell System Technical Journal}, vol.~28, no.~4, pp. 656--715, 1949.

\bibitem{wyner1975wire}
A.~D. Wyner, ``The wire-tap channel,'' \emph{Bell System Technical Journal}, vol.~54, no.~8, pp. 1355--1387, 1975.

\bibitem{hamamreh2018classifications}
J.~M. Hamamreh, H.~M. Furqan, and H.~Arslan, ``Classifications and applications of physical layer security techniques for confidentiality: A comprehensive survey,'' \emph{IEEE Commun. Surv. Tutor}, vol.~21, no.~2, pp. 1773--1828, 2018.

\bibitem{goel2008guaranteeing}
S.~Goel and R.~Negi, ``Guaranteeing secrecy using artificial noise,'' \emph{IEEE Trans. Wireless Commun}, vol.~7, no.~6, pp. 2180--2189, 2008.

\bibitem{akgun2016exploiting}
B.~Akgun, O.~O. Koyluoglu, and M.~Krunz, ``Exploiting full-duplex receivers for achieving secret communications in multiuser \uppercase{MISO} networks,'' \emph{IEEE Trans. Commun.}, vol.~65, no.~2, pp. 956--968, 2016.

\bibitem{siyari2017friendly}
P.~Siyari, M.~Krunz, and D.~N. Nguyen, ``Friendly jamming in a mimo wiretap interference network: A nonconvex game approach,'' \emph{J. Sel. Areas Commun.}, vol.~35, no.~3, pp. 601--614, 2017.

\bibitem{choi2015robust}
J.~Choi, ``A robust beamforming approach to guarantee instantaneous secrecy rate,'' \emph{IEEE Trans. Wireless Commun}, vol.~15, no.~2, pp. 1076--1085, 2015.

\bibitem{MIMO1}
D.~N. Nguyen and M.~Krunz, ``Spectrum management and power allocation in mimo cognitive networks,'' in \emph{2012 Proceedings IEEE INFOCOM}, 2012, pp. 2023--2031.

\bibitem{MIMO2}
------, ``Price-based joint beamforming and spectrum management in multi-antenna cognitive radio networks,'' \emph{IEEE Journal on Selected Areas in Communications}, vol.~30, no.~11, pp. 2295--2305, 2012.

\bibitem{MIMOML}
T.~X. Vu, S.~Chatzinotas, V.-D. Nguyen, D.~T. Hoang, D.~N. Nguyen, M.~D. Renzo, and B.~Ottersten, ``Machine learning-enabled joint antenna selection and precoding design: From offline complexity to online performance,'' \emph{IEEE Transactions on Wireless Communications}, vol.~20, no.~6, pp. 3710--3722, 2021.

\bibitem{mukherjee2015physical}
A.~Mukherjee, ``Physical-layer security in the internet of things: Sensing and communication confidentiality under resource constraints,'' \emph{Proceedings of the IEEE}, vol. 103, no.~10, pp. 1747--1761, 2015.

\bibitem{mukherjee2010robust}
A.~Mukherjee and A.~L. Swindlehurst, ``Robust beamforming for security in mimo wiretap channels with imperfect csi,'' \emph{IEEE Trans. Signal Process.}, vol.~59, no.~1, pp. 351--361, 2010.

\bibitem{6101597}
M.~Pei, J.~Wei, K.-K. Wong, and X.~Wang, ``Masked beamforming for multiuser mimo wiretap channels with imperfect csi,'' \emph{IEEE Trans. Wireless Commun}, vol.~11, no.~2, pp. 544--549, 2012.

\bibitem{9319238}
J.~B. Perazzone, P.~L. Yu, B.~M. Sadler, and R.~S. Blum, ``Artificial noise-aided mimo physical layer authentication with imperfect csi,'' \emph{IEEE Transactions on Information Forensics and Security}, vol.~16, pp. 2173--2185, 2021.

\bibitem{Tsai_Power_Allocation}
S.-H. Tsai and H.~V. Poor, ``Power allocation for artificial-noise secure mimo precoding systems,'' \emph{IEEE Trans. Signal Process}, vol.~62, no.~13, pp. 3479--3493, 2014.

\bibitem{Hu_Co_Jamming_PLS}
L.~Hu, H.~Wen, B.~Wu, F.~Pan, R.-F. Liao, H.~Song, J.~Tang, and X.~Wang, ``Cooperative jamming for physical layer security enhancement in internet of things,'' \emph{IEEE Internet of Things Journal}, vol.~5, no.~1, pp. 219--228, 2018.

\bibitem{8957321}
S.~Yun, J.-M. Kang, I.-M. Kim, and J.~Ha, ``Deep artificial noise: Deep learning-based precoding optimization for artificial noise scheme,'' \emph{IEEE Trans. Veh. Technol}, vol.~69, no.~3, pp. 3465--3469, 2020.

\bibitem{Guo_Proactive_multiagent}
D.~Guo, H.~Ding, L.~Tang, X.~Zhang, L.~Yang, and Y.-C. Liang, ``A proactive eavesdropping game in mimo systems based on multiagent deep reinforcement learning,'' \emph{IEEE Transactions on Wireless Communications}, vol.~21, no.~11, pp. 8889--8904, 2022.

\bibitem{erpek2018learning}
T.~Erpek, T.~J. O'Shea, and T.~C. Clancy, ``Learning a physical layer scheme for the mimo interference channel,'' in \emph{2018 IEEE International Conference on Communications (ICC)}.\hskip 1em plus 0.5em minus 0.4em\relax IEEE, 2018, pp. 1--5.

\bibitem{o2017deep}
T.~J. O'Shea, T.~Erpek, and T.~C. Clancy, ``Deep learning based \uppercase{MIMO} communications,'' \emph{arXiv preprint arXiv:1707.07980}, 2017.

\bibitem{fritschek2019deep}
R.~Fritschek, R.~F. Schaefer, and G.~Wunder, ``Deep learning for the gaussian wiretap channel,'' in \emph{ICC 2019-2019 IEEE International Conference on Communications (ICC)}.\hskip 1em plus 0.5em minus 0.4em\relax IEEE, 2019, pp. 1--6.

\bibitem{pmlr-v80-belghazi18a}
\BIBentryALTinterwordspacing
M.~I. Belghazi, A.~Baratin, S.~Rajeshwar, S.~Ozair, Y.~Bengio, A.~Courville, and D.~Hjelm, ``Mutual information neural estimation,'' in \emph{Proceedings of the 35th International Conference on Machine Learning}, ser. Proceedings of Machine Learning Research, J.~Dy and A.~Krause, Eds., vol.~80.\hskip 1em plus 0.5em minus 0.4em\relax PMLR, 10--15 Jul 2018, pp. 531--540. [Online]. Available: \url{https://proceedings.mlr.press/v80/belghazi18a.html}
\BIBentrySTDinterwordspacing

\bibitem{9449919}
N.~A. Letizia and A.~M. Tonello, ``Capacity-driven autoencoders for communications,'' \emph{IEEE Open Journal of the Communications Society}, vol.~2, pp. 1366--1378, 2021.

\bibitem{Fulwani_DL_MISO}
Y.~Fulwani, S.~Thapar, and N.~Sood, ``Deep learning based secure miso transmission,'' in \emph{2020 5th International Conference on Computing, Communication and Security (ICCCS)}, 2020, pp. 1--5.

\bibitem{lin2019beamforming}
T.~Lin and Y.~Zhu, ``Beamforming design for large-scale antenna arrays using deep learning,'' \emph{IEEE Wireless Communications Letters}, vol.~9, no.~1, pp. 103--107, 2019.

\bibitem{Fritschek2019}
R.~Fritschek, R.~F. Schaefer, and G.~Wunder, ``Deep learning for channel coding via neural mutual information estimation,'' \emph{arXiv preprint arXiv:1903.02865}, 2019.

\bibitem{o2017introduction}
T.~O’Shea and J.~Hoydis, ``An introduction to deep learning for the physical layer,'' \emph{IEEE Trans. Cogn. Commun. Netw.}, vol.~3, no.~4, pp. 563--575, 2017.

\bibitem{bengio2017deep}
Y.~Bengio, I.~Goodfellow, and A.~Courville, \emph{Deep learning}.\hskip 1em plus 0.5em minus 0.4em\relax MIT Proess, 2017, vol.~1.

\bibitem{agakov2004algorithm}
D.~B.~F. Agakov, ``The im algorithm: a variational approach to information maximization,'' 2004.

\bibitem{pmlr-v108-mcallester20a}
D.~McAllester and K.~Stratos, ``Formal limitations on the measurement of mutual information,'' in \emph{Procee. of the Twenty Third International Conference on Artificial Intelligence and Statistics}, ser. Proceedings of Machine Learning Research, S.~Chiappa and R.~Calandra, Eds., vol. 108.\hskip 1em plus 0.5em minus 0.4em\relax PMLR, 26--28 Aug 2020, pp. 875--884.

\bibitem{telatar1999capacity}
E.~Telatar, ``Capacity of multi-antenna gaussian channels,'' \emph{Eur. Trans. Telecommun.}, vol.~10, no.~6, pp. 585--595, 1999.

\bibitem{10075335}
F.~Mirkarimi, S.~Rini, and N.~Farsad, ``Benchmarking neural capacity estimation: Viability and reliability,'' \emph{IEEE Trans Commun}, vol.~71, no.~5, pp. 2654--2669, 2023.

\bibitem{9134818}
A.~Mohammad, C.~Masouros, and Y.~Andreopoulos, ``Complexity-scalable neural-network-based mimo detection with learnable weight scaling,'' \emph{IEEE Trans Commun}, vol.~68, no.~10, pp. 6101--6113, 2020.

\bibitem{9852985}
T.~V. Nguyen, D.~N. Nguyen, M.~D. Renzo, and R.~Zhang, ``Leveraging secondary reflections and mitigating interference in multi-irs/ris aided wireless networks,'' \emph{IEEE Transactions on Wireless Communications}, vol.~22, no.~1, pp. 502--517, 2023.

\bibitem{10437125825}
M.~Abughalwa, D.~Nguyen, D.~T. Hoang, and E.~Dutkiewic, ``Multi-user secrecy rate maximization in {IRS}-aided systems,'' \emph{to appear at IEEE Global Communications Conference (GLOBECOM)}, 2024.

\end{thebibliography}
\end{document}